\documentclass[a4paper,UKenglish,cleveref, autoref, thm-restate]{lipics-v2021}
\hideLIPIcs  

\usepackage{esvect}  
\usepackage{tikz}
\usepackage{tikz-cd}
\usepackage{color, soul}
\usepackage{thmtools}
\usepackage{wrapfig}
\usepackage{enumitem}

\newcommand*\ourCopyright[1]{%
  \bigskip
  \nobreakspace
  \scriptsize
  \begin{minipage}[t]{6em}
  \href{https://creativecommons.org/licenses/by/4.0/}%
       {\includegraphics[height=2.0em,clip]{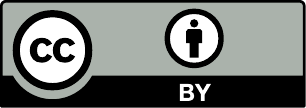}}
  \end{minipage}
  \begin{minipage}[t]{25em}
  \vskip-2.0em
  \textcopyright\ %
  #1
  ;\\%
  licensed under Creative Commons License CC-BY 4.0\par%
  \end{minipage}
  }

\usepackage{fancyhdr}
\fancypagestyle{plain}{%
  \fancyhf{} 
  \fancyfoot[L]{\ourCopyright{T.\ Jakl, D.\ Marsden, N.\ Shah}}
}

\newcommand{\nats}{\mathbb{N}}

\newcommand{\sg}{\sigma}

\newcommand{\const}[1]{\mathbf{cs}(#1)}
\newcommand{\rel}[1]{\mathbf{rel}(#1)}
\newcommand{\consts}[2]{\mathbf{cs}(#1,#2)}

\newcommand{\R}{\mathcal R}
\newcommand{\Rnc}{\mathcal R_{\text{nc}}}
\newcommand{\Rs}{\mathcal R^S}
\newcommand{\Rsnc}{\mathcal R^S_{\text{nc}}}
\newcommand{\Rel}{\mathcal R(\sg)}
\newcommand{\RelS}{\mathcal R^S(\sg)}
\newcommand{\RelE}{\mathcal R^S(\sgE)}
\newcommand{\RelI}{\mathcal R(\sgI)}
\newcommand{\RelJ}{\mathcal R^S(\sgJ)}

\newcommand{\sgm}{\sg_m}
\newcommand{\sgEm}{\sgE_m}

\newcommand{\Relk}{\mathcal R(\sgm)}
\newcommand{\RelEm}{\mathcal R^S(\sgEm)}

\newcommand\Gens{\mathcal{B}}
\newcommand\Ops{\mathcal{H}}

\newcommand\I{^I}
\newcommand\J{^J}
\newcommand\E{^e}
\newcommand\tr{\mathfrak t}
\newcommand\trI{\tr\I}
\newcommand\trE{\tr\E}
\newcommand\trEnc{\tr\E_\text{nc}}
\newcommand\trJ{\tr\J}
\newcommand\trJnc{\tr\J_\text{nc}}
\newcommand{\oE}{\mathfrak u\E}
\newcommand\sgI{\sg\I}
\newcommand\sgJ{\sg\J}
\newcommand\sgE{\sg\E}

\newcommand\opname[1]{\operatorname{#1}}

\newcommand{\rhoij}{\opname{relab}_{i,j}}
\newcommand{\etaRiPref}{\opname{join}}
\newcommand{\etaRiSuff}{_{\bar{i}}^R}
\newcommand{\etaRiSuffE}{_{\bar{i}}^{R,e}}
\newcommand{\etaRi}{{\etaRiPref\etaRiSuff}}
\newcommand{\etaRiE}{{\etaRiPref\etaRiSuffE}}

\newcommand\plus{\mathbin{\parallel}}

    \newcommand\inc{\mathfrak i} 
    \newcommand\ex{\mathtt{ex}}  
\newcommand\forg{\opname{fg}}
\newcommand\forgCD{\forg_{C,D}}

\newcommand{\logfrag}[1]{\text{#1}}
\newcommand{\MSO}{\logfrag{MSO}}

\newcommand{\FO}{\logfrag{FO}}

\newcommand{\CFO}{\#\FO}
\newcommand{\PE}{\exists^+}
\newcommand\fst{\mathbf{f}}  
\newcommand\snd{\mathbf{m}}  

\newcommand{\EM}[1]{\mathsf{EM}(#1)}
\newcommand{\EMF}[1]{F^{#1}}
\newcommand{\EMU}[1]{U^{#1}}
\newcommand{\Klei}[1]{\mathsf{Kl}(#1)}
\newcommand{\KleiEmpty}{\mathsf{Kl}}
\newcommand{\KLF}[1]{F_{#1}}

\newcommand{\Paths}{\mathcal{P}}

\newcommand{\embed}{\rightarrowtail}
\newcommand{\quot}{\twoheadrightarrow}
\newcommand\id{\mathrm{id}}
\newcommand{\Id}{{Id}}

\newcommand\counit{\varepsilon}
\newcommand\bisim{\leftrightarrow}
\newcommand{\Pw}{\mathcal P}

\newcommand{\lift}[1]{\widehat{#1}}          
\newcommand{\klift}[1]{\overline{#1}}        
\newcommand{\blift}[1]{\mathbin{\lift{#1}}}  
\newcommand{\kblift}[1]{\mathbin{\klift{#1}}} 
\newcommand\op{H}                
\newcommand\lop{\lift{\op}}      
\newcommand{\klop}{\klift{\op}}  
\newcommand{\lmulti}[1]{\wh{#1}}  

\newcommand\opveci[1]{\op(\vec{#1_i})}
\newcommand\opvec[1]{\op(\vec{#1})}

\newcommand\lopveci[1]{\lop(\vec{#1_i})}

\newcommand{\Fraisse}{Fra\"{i}\-ss\'{e}}
\newcommand{\ef}{Ehren\-feucht--\Fraisse}

\newcommand{\ol}{\overline}

\newcommand{\wh}{\widehat}
\newcommand{\sue}{\subseteq}

\newcommand\ee[1]{\enspace #1 \enspace}
\newcommand\qtq[1]{\quad\text{#1}\quad}

\newcommand\ete[1]{\ee{\text{#1}}}

\newcommand{\comonad}[1]{\mathbb{#1}} 
\newcommand\C{\mathbb C} 
\newcommand\D{\mathbb D}

\newcommand{\Ek}{\comonad{E}_{k}}
\newcommand{\Eknc}{\comonad{E}_{k}^{\text{nc}}}
\newcommand{\Pk}{\comonad{P}_{k}}
\newcommand{\Mk}{\comonad{M}_{k}}
    \newcommand\wrdk[1]{{#1}^{\leq k}}  

\newcommand{\cat}[1]{\mathcal{#1}}
\newcommand{\CA}{\cat{A}}
\newcommand{\CB}{\cat{B}}
\newcommand{\CC}{\cat{C}}
\newcommand{\CD}{\cat{D}}

\newcommand\Ec{\cat{E}}
\newcommand\Mc{\cat{M}}

\newcommand{\struct}[1]{#1}  
\newcommand{\As}{\struct{A}}
\newcommand{\Bs}{\struct{B}}
\newcommand{\Cs}{\struct{C}}
\newcommand{\Ps}{\struct{P}}
\newcommand{\Qs}{\struct{Q}}

\newcommand{\Ac}{\alpha}
\newcommand{\Bc}{\beta}

\newcommand{\Pc}{\pi}
\newcommand{\Qc}{\rho}

\newcommand{\setm}{\{1,\dots,m\}}
\newcommand{\univ}{\mathrm{u}}
\newcommand{\EF}{\mathbf{EF}}
\renewcommand\vec[1]{\vv{#1}}
\newcommand{\klcomp}{\bullet}

\newcommand{\df}[1]{\emph{#1}} 

\newcommand\TODO[1][]{\noindent{\color{blue}\textbf{TODO} #1}}

\newcommand\hides[1]{}

\newenvironment{axioms}{\begin{enumerate}[labelsep=8pt,leftmargin=*,itemindent=2em,labelindent=0.8\parindent]}{\end{enumerate}}

\makeatletter
\def\namedlabel#1#2{\begingroup
    #2%
    \def\@currentlabel{#2}%
    \phantomsection\label{#1}\endgroup
}
\makeatother


\title{A game comonadic account of Courcelle and Feferman--Vaught--Mostowski theorems}
\titlerunning{A game comonadic account of Courcelle and Feferman-Vaught-Mostowski theorems.}

\author{Tom\'a\v s Jakl}{University of Cambridge, \url{https://tomas.jakl.one}}{tomas.jakl@cl.cam.ac.uk}{https://orcid.org/0000-0003-1930-4904}{}
\author{Dan Marsden}{University of Oxford}{daniel.marsden@cs.ox.ac.uk}{https://orcid.org/0000-0003-0579-0323}{}
\author{Nihil Shah}{University of Oxford}{nihil.shah@cs.ox.ac.uk}{https://orcid.org/0000-0003-2844-0828}{}
\authorrunning{T.\ Jakl, D.\ Marsden and N.\ Shah}
\Copyright{T.\ Jakl, D.\ Marsden and N.\ Shah}
\nolinenumbers

\funding{Research supported by the EPSRC project EP/T00696X/1 ``Resources and Co-Resources: a junction between categorical semantics and descriptive complexity''.}


\begin{document}

\maketitle

\begin{abstract}
    Game comonads, introduced by Abramsky, Dawar and Wang, and developed by Abramsky and Shah, give a categorical semantics for model comparison games.
    We present an axiomatic account of Feferman-Vaught-Mostowski (FVM) composition theorems within the game comonad framework, parameterized by the model comparison game. In a uniform way, we produce compositionality results for the logic in question, and its positive existential and counting quantifier variants.
    
    Secondly, we extend game comonads to the second order setting, specifically in the case of Monadic Second Order (MSO) logic. We then generalize our FVM theorems to the second order case.
    We conclude with an abstract formulation of Courcelle's algorithmic meta-theorem, exploiting our earlier developments. This is instantiated to recover well-known bounded tree-width and bounded clique-width Courcelle theorems for MSO on graphs.
\end{abstract}

\section{Introduction}

Model comparison games, such as the \ef{}~\cite{ehrenfeucht1961application, fraisse1955quelques}, pebbling~\cite{kolaitis1992infinitary} and bisimulation games~\cite{hennessy1980observing}, play a key role in finite model theory~\cite{libkin2004elements}.
Recent work has introduced a novel categorical semantics for these games~\cite{AbramskyDW17, abramsky2021relating}. In this programme, a comonad is introduced for a chosen model comparison game. This comonad is typically graded by a resource parameter, such as quantifier depth or variable count. Winning Duplicator strategies for model comparison games, bounded by the chosen resource parameter, can then be uniformly characterized categorically, in terms of the Kleisli~\cite{kleisli1965every} and Eilenberg-Moore~\cite{eilenberg1965adjoint} categories of the comonad.

This approach has been extended to more complex games, such as those for guarded logics~\cite{AbramskyM21} and generalized quantifiers~\cite{ConghaileD21}, and applications such as Lovasz type theorems~\cite{DawarJR21} and Rossman's homomorphism theorem~\cite{Paine20}. The abstract setting was axiomatized in~\cite{AbramskyR21}.

Our aim is to give a semantic account of algorithmic meta-theorems, specifically Courcelle type theorems~\cite{courcelle2012graph}, exploiting the game comonadic semantics. 
To do so, firstly we develop an axiomatic formulation of Feferman-Vaught-Mostowski (FVM) theorems~\cite{mostowski1952direct, feferman1967first}, which describe how logical equivalence is preserved under operations on models. 

Pleasingly, the canonical categorical structures associated with the comonadic approach play a central role.
Specifically, we lift model transformations to the Kleisli and Eilenberg-Moore categories via distributive laws satisfying natural axioms. These axioms are semantic in nature, and parametric in the notion of model equivalence. Contrast this with the more tautological, MSO specific logical condition of MSOL-smoothness used to describe well-behaved model transformations for the generic Courcelle theorem presented in~\cite{makowsky2004algorithmic}. Our results yield FVM results for a given logic, it's positive existential fragment, and its extension to counting quantifiers in a uniform way.

Secondly, we identify a suitable comonad capturing the model comparison game for monadic second order logic. FVM theorems often generalize to MSO~\cite{gurevich1985monadic, gurevich1979modest, Lauchli1966, shelah1975monadic}. Our abstract approach to FVM theorems is then extended to the second-order setting.

With these two components in place, we give an abstract Courcelle theorem, parameterized by the game comonad and model transformations of interest. This theorem is then instantiated to yield concrete results for bounded tree-width~\cite{courcelle1990monadic} and bounded clique-width variants~\cite{courcelle2000linear, oum2006approximating} of Courcelle's theorem.

\section{Preliminaries}

We assume the reader is already familiar with basic category theoretic notions such as functors, natural transformations, adjunctions, and limits (see e.g.~\cite{abramskytzevelekos2010introduction} or \cite{awodey2010category}).

\subsection{Game comonads}

A \df{comonad} (in Kleisli form) on a category $\CC$ is a triple $(\C, \counit, (-)^*)$ where $\C$ is an object map $\mathrm{obj}(\CC)\to \mathrm{obj}(\CC)$, the \df{counit} $\counit$ is a~$\mathrm{obj}({\CC})$-indexed family of morphisms $\counit_A\colon \C(A) \to A$. Lastly, there is a coextension operation~$(-)^*$ mapping morphisms~$f: \C(\As) \rightarrow \Bs$ to their coextension~$f^*:\C(\As) \rightarrow \C(\Bs)$,
subject to the following equations:
\begin{align}
    (\counit_A)^* = \id_{\C A}, \quad \counit_B \circ f^* = f,\quad (g \circ f^*)^* = g^* \circ f^*,
    \label{eq:comonad-axioms}
\end{align}
It is then standard~\cite{manes2012algebraic} that $\C$ is a functor such that~$\C(f) = (f \circ \counit_\As)^*$, and the counit is a natural transformation~$\C \Rightarrow \Id$. The morphisms~$\delta_\As := \id_{\C(A)}^* : \C(A) \rightarrow \C^2(A)$ form a \df{comultiplication} natural transformation satisfying
$\counit \circ \delta = \id = \C(\counit) \circ \delta$ and $\delta \circ \delta = \C(\delta) \circ \delta$.

Our comonads of interest shall be on the category of relational structures over a fixed signature. A~\df{signature}~$\sg$ is a set consisting of:
\begin{enumerate}
    \item \df{Relation symbols}, each with a specified \df{arity}, a strictly positive natural number. A relation symbol of arity~$n$ is referred to as \df{$n$-ary}.
    $\rel{\sg}$ shall denote the subset of relation symbols.
    \item \df{Constant symbols}. $\const{\sg}$ shall denote the subset of constant symbols.
\end{enumerate}
A~\df{modal signature}~$\sg$ has a single constant~$c_0$ and~$\rel{\sg}$ consisting of a single binary relation~$E$, and finitely many unary relations.

A~$\sg$-structure~$\As$ consists of a set~$A$, referred to as the~\df{universe}. For each $R \in \rel{\sg}$ of arity~$n$, an~$n$-ary relation~$R^\As$ on~$A$, and for each~$c \in \const{\sg}$ an element~$c^\As \in A$. A \df{morphism of~$\sg$-structures} $f : \As \rightarrow \Bs$ is a function of type~$A \rightarrow B$, preserving relations and constants. That is, for~$n$-ary $R \in \rel{\sg}$, $R^\As(a_1,\ldots,a_n) \Rightarrow R^\Bs(f(a_1),\ldots, f(a_n))$, and for $c \in \const{\sg}$, $ f(c^\As) = c^\Bs$.
$\Rel$ shall denote the category of $\sg$-structures and morphisms between them.

We shall require three comonads on~$\Rel$ for the examples in later sections, parameterized by a positive integer~$k$. These are small generalizations of the key comonads appearing in~\cite{abramsky2021relating}, to handle signatures with constants in a uniform way.

We shall write~$A^+$ for the set of non-empty lists over~$A$, and~$A^{\leq k}$ for its restriction to lists of at most~$k$ elements.
\begin{description}
    \item[Comonad~$\Ek$:] The universe of~$\Ek(\As)$ is~$\As^{\leq k} \uplus \const{\sg}$. For~$c \in \const{\sg}$, define~$c^{\Ek(\As)} := c$. The counit preserves constants, and on lists $ \counit[a_1,\dots,a_n] := a_n $.
    For~$n$-ary~$R \in \rel{\sg}$, $R^{\Ek(\As)}(s_1,\ldots,s_n)$ iff the $s_i \not\in \const{\sg}$ are pairwise comparable in the prefix order, and
    $R^\As(\counit(s_1),\ldots,\counit(s_n))$.
    For~$h : \Ek(\As) \rightarrow \Bs$, the coextension preserves constants, and acts on lists as
        $h^*([a_1,\ldots,a_n]) := [h([a_1]), \ldots, h([a_1,\ldots,a_n])]$.
    \item[Comonad~$\Mk$:] We restrict to a modal signature. The universe of~$\Mk(\As)$ consists of~$\{ c_0 \}$ and the elements of~$\As^{\leq k}$ of the form~$[a_1,\ldots,a_n]$ such that~$E^{\As}(c_0^\As,a_1)$ and for~$1 \leq i < n$, $E^{\As}(a_i, a_{i + 1})$. The counit and coextension act as for~$\Ek$.
    We have~$E^{\Mk(\As)}(s,s')$ iff either~$s = c_0$ and $s'$ is a one element list, or~$s$ and~$s'$ are lists with~$s'$ covering~$s$ in the prefix order. For unary $P \in \rel{\sg}$, $P^{\Mk(\As)}(s)$ iff $P^{\As}(\counit(s))$.
     
    \item[Comonad~$\Pk$:] The universe of $\Pk(\As)$ is $(\{ 0,\ldots,k - 1 \} \times \As)^+ \,\uplus\, \const{\sg}$. 
    Constant interpretations are defined as for~$\Ek$. The counit $\counit$ preserves constants, and on lists
    sends $[(p_1,a_1),\dots,(p_n,a_n)]$ to $a_n$.
    We refer to the first element in a pair~$(p,a)$ as a~\df{pebble index}.  For $n$-ary~$R \in \rel{\sg}$,
    $R^{\Pk(\As)}(s_1,\ldots,s_n)$ iff:
    \begin{enumerate}
        \item The $s_i \not\in \const{\sg}$ are pairwise comparable in the prefix order.
        \item For~$s_i, s_j \not\in \const{\sg}$ with~$s_i \leq s_j$, the pebble index of the tail of~$s_i$ does not appear in the suffix of~$s_i$ in~$s_j$.
        \item $R^{\As}(\counit(s_1),\ldots,\counit(s_n))$.
    \end{enumerate}
    For~$h\colon \Pk(\As) \rightarrow \Bs$, the coextension $h^*$ preserves constants, and acts on lists as:
    \[
        h^*([(p_1,a_1),\ldots,(p_n,a_n)]) := [(p_1,h([a_1])), \ldots, (p_n, h([a_1,\ldots,a_n]))]
    \]
\end{description}
\begin{restatable}{theorem}{StandardComonadsWithConstants}
$\Ek$, $\Mk$ and~$\Pk$ are comonads in Kleisli form \ (direct adaptation of \cite{AbramskyDW17,abramsky2021relating}). 
\end{restatable}

\subsection{Kleisli and Eilenberg--Moore categories}

A comonad $\C$ on a category $\CC$ gives rise to two categories. The first is the \df{Kleisli category} $\Klei{\C}$ consisting of the same objects as $\CC$ has, but morphisms from $A$ to $B$ are $\CC$-morphisms of type $\C(A) \to B$. The composition of $f\colon \C A \to B$ and $g\colon \C B \to C$ is defined as $g \circ f^*$, with the counits the identity morphisms. The equations \eqref{eq:comonad-axioms} ensure that $\Klei{\C}$ is a well-formed category.

The second category induced by $\C$ is the \df{Eilenberg--Moore category~$\EM{\C}$ of $\C$-coalgebras}. Its objects are pairs $(A,\alpha)$ where $\alpha\colon A \to \C(A)$ is a $\CC$-morphism such that
\begin{equation}
    \counit_A \circ \alpha = \id_A \ete{and} \delta \circ \alpha = \C \alpha \circ \alpha,
    \label{eq:coalgebra-axioms}
\end{equation}
A \df{coalgebra morphism} $h\colon (A,\alpha)\to (B,\beta)$ is a $\CC$-morphism $h\colon A\to B$ such that $\beta\circ h = \C h \circ \alpha$. We shall simply write $\alpha$ instead of $(\As,\alpha)$ whenever~$\As$ can be inferred from the context.

\begin{wrapfigure}[4]{r}{17em}
    \vspace{-1em}
    \centering
    \begin{tikzcd}[column sep=4em]
        \Klei{\C}
            \ar[bend right=18,swap, shorten <= 0.5em, shorten >= 0.3em]{dr}{U_\C}
            \ar[bend left=10]{rr}{K^\C}
        &
        & \EM{\C}
            \ar[bend right=18,swap]{dl}{U^\C}
        \\
        & \CC
            \ar[bend right=18,swap]{ul}{\KLF{\C}}
            \ar[phantom,sloped,pos=0.55]{ul}{\rotatebox{80}{$\dashv$}}
            \ar[bend right=18,swap, shorten <= 0.5em, shorten >= 0.3em]{ur}{\EMF{\C}}
            \ar[phantom,sloped]{ur}{\rotatebox{-80}{$\dashv$}}
    \end{tikzcd}
\end{wrapfigure}
The two categories induce two adjunctions over the base category and a functor $K^\C$ from $\Klei{\C}$ to $\EM{\C}$, as depicted on the right.
For our purposes it suffices to know that
\begin{itemize}
    \item $K^\C$ is a fully faithful functor such that $K^\C \circ \KLF{\C} = \EMF{\C}$.
    \item $F_\C$ is the identity on objects, and for~$\CC$-morphism $f\colon \As \to \Bs$, $\KLF{\C}(f) = f\circ \counit_A$.
    \item $U^\C$ is the forgetful functor, sending $(\As,\alpha)$ to $\As$ and $\EMF{\C}$ is the \df{cofree coalgebra} functor, sending $\As$ in $\CC$ to $(\C \As, \delta_{\As})$ and $f\colon \As \to \Bs$ in $\CC$ to $\C f$.
\end{itemize}

\begin{remark}
    A coalgebra $(A,\alpha)$ of any of the game comonads defined above is endowed with a preorder~$\sqsubseteq_\alpha$, where $x \sqsubseteq_\alpha y$ whenever the word $\alpha(x)$ is a prefix of $\alpha(y)$.
    Then, by \eqref{eq:coalgebra-axioms}, $\alpha(y)$, for a non-constant $y$, is the word listing all non-constant $x$ such that $x \sqsubseteq_\alpha y$.

    In fact, the category $\EM{\Ek}$ is equivalent to the category of $\sg$-structures endowed with a compatible forest order of depth ${\leq}\,k$, see section 9 in \cite{abramsky2021relating} for details.
    \label{r:coalg-order}
\end{remark}

\subsection{Factorisation systems}
\label{s:fact-syst}
\df{Strong homomorphisms} in~$\Rel$ are those homomorphisms $f\colon A\to B$ that reflect relations, i.e.\ if $(f(x_1),\dots, f(x_n))\in R^B$ then $(x_1,\dots,x_n)\in R^A$.
Recall that any homomorphism $h\colon A\to B$ in $\Rel$ factors into an onto homomorphism $q\colon A \quot h[A]$ followed by a strong injective homomorphism $m\colon h[A] \embed B$.
Then, $h[A]$ above is the $\sg$-structure on the image $h[A]$ of $h$ so that the inclusion $h[A] \to B$ is a strong homomorphism. This factorisation is a special case of a general categorical structure which we review next.

A \df{factorisation system} $(\Ec, \Mc)$ consists of a class of morphisms $\Ec$ called \df{quotients}, denoted by $\quot$, and a class of morphisms $\Mc$ called \df{embeddings}, denoted by $\embed$, which are both closed under compositions, contain all isomorphisms and, furthermore,
\begin{itemize}
    \item every morphism $f\colon A \to B$ has a factorisation $f = m \circ e$ for some quotient $e\colon A \quot C$ and embedding $m\colon C\embed B$, and
    \item
    \parbox[t]{\dimexpr\textwidth-\leftmargin}{%
      \vspace{-2.5mm}
    \begin{wrapfigure}[3]{r}{7em}
      \vspace{-2.2em}
    \centering
    $
        \begin{tikzcd}[ampersand replacement=\&]
            A \rar[->>]{e}\dar[swap]{u} \& B \dar{v} \ar[swap,dashed]{ld}{d} \\
            C \rar[>->]{m} \& D
        \end{tikzcd}
    $
    \end{wrapfigure}
    every $e \in \Ec$ is \df{orthogonal} to every $m \in \Mc$, i.e.\ given $u,v$ making the diagram on the right commute, 
    there exists a unique morphism $d$, making the two triangles commute.
      \vspace{0.5em}
    }
\end{itemize}
If, furthermore, $\Ec$ is a class of epimorphisms and $\Mc$ is a class of monomorphisms, we say that $(\Ec,\Mc)$ is a \df{proper factorisation system}.

We make use of the following well-known fact that any comonad $\C$ which \df{restricts to embeddings}, i.e.\ sends $\Mc$-morphisms to $\Mc$-morphisms, lifts $(\Ec,\Mc)$ to the factorisation system $(\ol\Ec, \ol\Mc)$ on $\EM{\C}$, where $\ol\Ec$ consists of all morphisms of coalgebras $h$ such that $U(h)$ is in $\Ec$ and, similarly, $\ol\Mc$ consists of morphisms $h$ such that $U(h)$ is in $\Mc$.

\begin{restatable}{lemma}{EMFactorization}
    $(\ol\Ec, \ol\Mc)$ is a factorisation system and it is proper whenever $(\Ec, \Mc)$ is.
    \label{l:lifting-EMfs}
\end{restatable}

In the following we always assume that categories of relational structures $\Rel$ come equipped with the (surjective, strong injective) factorisation system\footnote{Categorically, surjective homomorphisms are epis and strong injective homomorphisms are regular monos in $\Rel$.}. It is immediate to see that $\Ek$ preserves strong injective homomorphisms and, hence, by the previous lemma, this factorisation system lifts to the categories of $\Ek$-coalgebras. Moreover, the same is true for pebbling and modal comonads and their categories of coalgebras.

\subsection{Game comonads and logical equivalences}

An important feature of game comonads is that they semantically characterise logical equivalences with respect to various fragments and extensions of first-order logic. Let $\FO_k$ be the fragment of first-order logic over signature $\sg$ restricted to sentences of quantifier rank $k$, i.e.\ when the depth of nesting of quantifiers is at most $k$. Then, $\PE_k$ is the further refinement of $\FO_k$ to positive existential sentences, i.e.\ sentences without negation or universal quantification, and $\#\FO_k$ is the extension of $\FO_k$ by allowing counting quantifiers. These are quantifiers $\exists_{\geq n}$ where $A \models \exists_{\geq n} x. \varphi(x)$ if there are at least $n$ different $a\in A$ such that $A\models \varphi(a)$.

As we will see, the \ef{} comonad $\Ek$ characterises logical equivalence with respect to all three aforementioned fragments. To this end, define
\begin{itemize}
    \item $\As \leftrightarrows_{\C} \Bs$ if there exist morphisms $\C(\As) \to \Bs$ and $\C(\Bs) \to \As$, and
    \item $\As \cong_{\Klei\C} B$ if $\KLF{\C}(\As) \cong \KLF{\C}(\Bs)$ in $\Klei{\C}$.
\end{itemize}
In fact, these equivalences are all statements about free coalgebras as $A \cong_{\Klei\C} B$ is equivalent to $\EMF{\C}(\As) \cong \EMF{\C}(\Bs)$ in $\EM{\C}$, since $K^\C$ is fully faithful, and similarly, $\As \leftrightarrows_{\C} \Bs$ iff some $\EMF{\C}(\As) \to \EMF{\C}(\Bs)$ and $\EMF{\C}(\As) \to \EMF{\C}(\Bs)$ exist.

Following~\cite{AbramskyDW17}, to capture logics with equality we introduce an extended signature~$\sgI$ with a fresh binary relation~$I$, and a~\df{translation functor} $\trI\colon \R_0(\sg) \to \RelI$
sending $A\in \R_0(\sg)$ to $(A,I^A)$, where $(a,b)\in I^A$ iff $a = b$. Here, $\R_0(\sg)$ denotes the discrete subcategory of $\Rel$ containing only identity morphisms. Since $\Ek$ was defined for an arbitrary relational signature, there is an instance of $\Ek$ over $\RelI$, which we don't distinguished notationally.

\begin{proposition}[by adapting~{\cite[theorems 13 and 16]{AbramskyDW17}}]
    \label{p:ef-logic-1}
    For $\sg$-relational structures $A,B$,
    \begin{itemize}
        \item $A \equiv_{\PE_k} B$  \ee{iff} $A \leftrightarrows_{\Ek} B$ \ee{iff} $\trI A \leftrightarrows_{\Ek} \trI B$
        \item $A \equiv_{\CFO_k} B$ \ee{iff} $\trI A \cong_{\Klei\Ek} \trI B$
    \end{itemize}
\end{proposition}

Capturing $\equiv_{\FO_k}$ requires a few more definitions. First, the category of \df{paths} $\Paths$ is the full subcategory of $\EM{\Ek}$ consisting of the coalgebras $(A,\alpha)$ such that the order $\sqsubseteq_\alpha$ (cf.\ remark~\ref{r:coalg-order}) is a finite linear order. Recall from section~\ref{s:fact-syst} that the factorisation system (surjective, strong injective) on $\Rel$ lifts to the category of $\Ek$-coalgebras.
We call a morphism $h\colon (A,\alpha) \to (B,\beta)$ in $\EM{\Ek}$ a \df{pathwise-embedding}, if for any path embedding
\begin{wrapfigure}[4]{r}{9em}
    \vspace{-1em}
    \centering
    $
    \begin{tikzcd}
    (P,\pi) \rar[>->] \dar[>->] & (Q,\rho)\dar[>->]\\
    (A,\alpha) \rar{h} & (B,\beta)
    \end{tikzcd}
    $
\end{wrapfigure}
$e\colon (P,\pi)\embed (A,\alpha)$ (with $(P,\pi)\in \Paths$),
the composite $h\circ e$ is also an embedding.
Further, we say that $h$ is \df{open} if every commutative square as shown on the right,
where $(P,\pi)$ and $(Q,\rho)$ are from $\Paths$, has a diagonal filler $d\colon (Q,\rho) \to (A,\alpha)$ making the two triangles commute. Finally, define
\begin{itemize}
    \item $A \bisim_{\C} B$ if there exists a span of open pathwise-embeddings $\EMF{\C}(A) \leftarrow R \rightarrow \EMF{\C}(B)$. 
\end{itemize}

\begin{proposition}[by adapting~{\cite[theorem 10.5]{abramsky2021relating}}]
    \label{p:ef-logic-2}
    $A \equiv_{\FO_k} B$  \ee{iff} $\trI A \bisim_{\Ek} \trI B$. 
\end{proposition}

The pebbling comonad $\Pk$ captures the corresponding $k$-variable fragments (as opposed to quantifier rank fragment) in the same way and, similarly, $\Mk$ captures the fragments of modal logic restricted to formulas of modal depth ${\leq}\,k$.

\begin{remark}
    \label{r:coalg-order-abstract}
    We always assume that the subcategory of paths is given implicitly.
    In particular, for any category $\CC$ equipped with a factorisation system $(\Ec,\Mc)$, we take the full subcategory of paths $\Paths \subseteq \CC$ to be the objects whose class of non-isomorphic subobjects (w.r.t $\Mc$-morphisms) forms a finite chain under the order induced by $\Mc$-morphisms. In particular, for the game comonads we discuss, the subcategory of $\Paths \subseteq \EM{\C}$ is equivalent to subcategory of the coalgebras $(\As,\alpha)$ where the order $\sqsubseteq_\alpha$ is a finite linear order ~\cite{AbramskyR21}. 
\end{remark}

\section{Comonadic FVM Theorems}

\subsection{A Motivating FVM-type Example}
\label{sec:motivating}
In the approach to algorithmic meta-theorems such as Courcelle's theorem adopted in later sections,
we work with classes of structures built up by a suitable family of operations. These operations must interact well
with logical equivalence in the sense of Feferman--Vaught--Mostowski. 
We now outline a simple motivating example, forming coproducts of structures,
and the relationship to first order logic. We use the \ef{} comonad, highlighting the key themes, and emphasizing connections to game comonads. The finer technical details are developed in later sections.

Fix a relational signature~$\sg$ without constants. For two structures~$\As$ and~$\Bs$, their coproduct~$\As + \Bs$ is simply
the disjoint union of the two components, with the induced relational structure.
We would like to show that, if $A_1 \equiv_{\FO_k} B_1$ and $A_2 \equiv_{\FO_k} B_2$, then also $\As_1 + \As_2 \equiv_{\FO_k} B_1+ B_2$. Equivalently, in view of proposition~\ref{p:ef-logic-2}, given spans of open pathwise-embeddings
$
    \EMF{\Ek}(\trI(\As_i)) \xleftarrow{f_i} R_i \xrightarrow{g_i} \EMF{\Ek}(\trI(\Bs_i)) 
$ for~$i \in \{1,2\}$,
we would like to construct a span of open pathwise embeddings of the form
\begin{align}
    \EMF{\Ek}(\trI(\As_1 + \As_2)) \leftarrow R \rightarrow \EMF{\Ek}(\trI(\Bs_1 + \Bs_2)).
    \label{eq:coprod-span-2}
\end{align}
Observe that, even though the category of coalgebras has coproducts for general reasons, a span of the form in~\eqref{eq:coprod-span-2} cannot be formed using the coproduct of the two component spans. We would obtain a span of type
\[ \EMF{\Ek}(\trI(\As_1)) + \EMF{\Ek}(\trI(\As_2)) \xleftarrow{f_1+f_2} R_1 + R_2 \xrightarrow{g_1+g_2} \EMF{\Ek}(\trI(\Bs_1)) + \EMF{\Ek}(\trI(\Bs_2)) \]
In general $\EMF{\Ek}(\trI(\As_1)) + \EMF{\Ek}(\trI(\As_2))$ and $\EMF{\Ek}(\trI(\As_1 + \As_2))$ need not be isomorphic, so our span is of the wrong type. Conceptually, this approach fails because the span we require should encode a strategy for play that may switch back and forth between two components, and forming coproducts keeps all the data entirely independent.

We observe that if we have a natural transformation
$\kappa\colon \Ek(\As + \Bs) \to \Ek(\As) + \Ek(\Bs) $
that satisfies some natural equations
with respect to the counit and comultiplication, then this
induces a lifted functor~$\kblift{+} : \Klei{\Ek} \times \Klei{\Ek} \rightarrow \Klei{\Ek}$ such that~$\KLF{\Ek}(\As + \Bs) = \KLF{\Ek}(\As) \kblift{+} \KLF{\Ek}(\Bs)$.
As~$\KLF{\Ek}(A) = A$, the existence of such a natural transformation yields both:
\begin{align*}
    \As_1 \equiv_{\PE_k} \As_2 \;\mbox{ and }\; \Bs_1 \equiv_{\PE_k} \Bs_2 \;&\Rightarrow\; \As_1 + \Bs_1 \equiv_{\PE_k} \As_2 + \Bs_2\\
    \As_1 \equiv_{\CFO_k} \As_2 \;\mbox{ and }\; \Bs_1 \equiv_{\CFO_k} \Bs_2 \;&\Rightarrow\; \As_1 + \Bs_1 \equiv_{\CFO_k} \As_2 + \Bs_2
\end{align*}
So we have FVM type results for the positive existential fragment, and~$\FO$ with counting quantifiers. For full first-order logic, we require an additional observation. As~$\EM{\Ek}$ has equalizers, we can lift coproducts to a functor~$\EM{\Ek} \times \EM{\Ek} \rightarrow \EM{\Ek}$ by forming the equalizer:
\begin{equation*}
    \begin{tikzcd}
        \As \blift+ \Bs
            \rar
        & \EMF{\Ek}(\As + \Bs)
            \ar[yshift=0.5em]{rr}{\EMF{\Ek}(\kappa)\circ \delta}
            \ar[swap,yshift=-0.5em]{rr}{\EMF{\Ek}(\alpha + \beta)}
        &
        & \EMF{\Ek}(\Ek(\As) + \Ek(\Bs))
    \end{tikzcd}
    \label{eq:plus-lift}
\end{equation*}
Crucially, this functor satisfies~$\EMF{\Ek}(\As + \Bs) \cong \EMF{\Ek}(\As) \blift{+} \EMF{\Ek}(\Bs)$. We can then verify the following by direct calculation.
\begin{restatable}{lemma}{InterleavingSumsPreserveOPE}
    \label{l:coprod-ope-preservation}
    Given a pair of open pathwise-embeddings $f,g$ in $\EM{\Ek}$, the morphism $f\blift+ g$ is also an open pathwise-embedding.
\end{restatable}
We see that the following span
\[ \EMF{\Ek}(\trI(\As_1)) \blift+ \EMF{\Ek}(\trI(\As_2)) \xleftarrow{f_1\blift+f_2} R_1 \blift+ R_2 \xrightarrow{g_1\blift+g_2} \EMF{\Ek}(\trI(\Bs_1)) \blift + \EMF{\Ek}(\trI(\Bs_2)) \]
is a span of open pathwise-embeddings, by lemma~\ref{l:coprod-ope-preservation}. Moreover, $\EMF{\Ek}(\trI(\As_1)) \blift+ \EMF{\Ek}(\trI(\As_2))$ is isomorphic to $\EMF{\Ek}(\trI(\As_1) + \trI(\As_2))$ and
similarly for $B_1$ and $B_2$,
and~$\trI(A) + \trI(B) \cong \trI(A + B)$,
yielding a span of shape \eqref{eq:coprod-span-2}, and so $\As_1 + \As_2 \equiv_{\FO_k}\Bs_1 + \Bs_2$.

\begin{remark}
    The above analysis follows from a well-known theory of commutative (resp.\ lax monoidal) monads~\cite{kock1970monads}. They lift the monoidal structure to the Kleisli and Eilenberg-Moore categories (if sufficiently complete~\cite{day1970closed, kock1971closed}). The liftings of coproducts in this section is arising from the dual of these results, as~$\Ek$ is opmonoidal with respect to coproducts.
\end{remark}

Motivated by this example, in the sequel we introduce a general framework for lifting logical equivalences. The general results deal with operations of arbitrary finite arity, with unary and binary transformations being of primary interest in applications. We will also need to deal with varying base categories and comonads, as signatures, resource parameters or even logics may vary under model transformations. As was the case in this example, this leads to FVM theorems for a given logic, its positive existential fragment, and its extension with counting quantifiers, parametric in the game comonad of interest. The main technical challenge is giving sufficient conditions for the analogue of the critical lemma~\ref{l:coprod-ope-preservation} to hold.

\subsection{FVM-Theorems from Kleisli-laws}
\label{s:fvm-thm-from-klei-laws}
As we saw in section~\ref{sec:motivating}, in order to develop FVM-type theorems within the comonadic framework, it is natural to consider functorial transformations of models.
We consider categories $\CC_1$, \dots, $\CC_n$ and $\CD$, equipped with comonads $\C_1,\dots,\C_n$, and $\D$, respectively, and a functorial operation $\op\colon \prod_i \CC_i \to \CD$. 

\begin{remark}
To simplify our notation for objects and morphisms in product categories, we shall write $\vec{A_i}$ and $\vec{f_i}$ instead of the longer expressions $A_1,\dots,A_n$ and $f_1,\dots,f_n$, respectively.
\end{remark}

A \df{Kleisli-law} is a natural transformation $ \kappa\colon \D \circ H \to \op \circ \prod_i \C_i $
such that:
\begin{axioms}
    \item[\textbf{\namedlabel{ax:kl-law-counit}{(F1)}}] $\op(\vec{\counit_i}) \circ \kappa = \counit_{\op}$
    \item[\textbf{\namedlabel{ax:kl-law-comultiplication}{(F2)}}] $\op(\vec{\delta_i}) \circ \kappa = \kappa \circ \D\kappa \circ \delta_{\op}$
\end{axioms}

These laws are direct generalisations of Kleisli-laws lifting of unary functor to the Kleisli category, see e.g. \cite{manes2007monad}. The following result is probably folklore~\cite{jacobs1994semantics}:
\begin{restatable}{proposition}{KleisliCorrespondence}
There is a bijective correspondence between:
\begin{enumerate}
    \item Kleisli-laws of type
    $ \D \circ \op \to \op \circ \prod_i \C_i $.
    \item Functors~$\klop : \prod_i \Klei{\C_i} \rightarrow \Klei{\D}$ such that:
    $\klop \circ \prod_i \KLF{\C_i} = \KLF{\D} \circ \op$.
    \label{eq:kleisli-lift}
\end{enumerate}
\end{restatable}
It is encouraging for the game comonad approach that this classical theory immediately yields results for logical equivalences. 
To this end, let $\approx_{\C_i}$ and~$\approx_{\D}$ be binary relations on~$\CC_i$ and~$\CD$ objects respective. We shall say that \df{$\op$ preserves~$\approx$} if for all~$1 \leq i \leq n$, $\As_i \mathbin{\approx_{\C_i}}\Bs_i$, then 
$\op(\As_1,\ldots, \As_n) \mathbin{\approx_{\D}} \op(\Bs_1,\ldots,\Bs_n)$.

\begin{restatable}{theorem}{KleisliLogicalEquivalences}
    If $\kappa : \D \circ \op \rightarrow \op \circ \prod_i \C_i$t is a Kleisli-law then~$\op$ preserves both~$\leftrightarrows$ and~$\cong_{\KleiEmpty}$.
    \label{t:FVM-from-kleisli-laws}
\end{restatable}

\begin{example}[Comonad Morphisms]
    \label{e:comonad-morphisms}
    A \df{comonad morphism}~$\sigma : \C \Rightarrow \D$ is a natural transformation such that
    $\counit \circ \sigma = \counit$ and $\delta \circ \sigma = \D(\sigma) \circ \sigma \circ \delta$.
    A comonad morphism is a special case of a Kleisli-law, of type~$\C \circ \Id \Rightarrow \Id \circ \D$, corresponding to liftings of the identity functor. 
    For example, there is a comonad morphism of type~$\comonad{M} \rightarrow \comonad{P}_2$, with action preserving constants, and mapping lists to lists pebbled by the parity of their position:
    \[ [a_1,a_2,a_3\ldots,a_n] \mapsto [(1,a_1),(0,a_2),(1,a_3)\ldots,(p_n,a_n)] \]
    where $p_n = n \mod 2$. 
    Consequently, utilizing the logical characterizations  analogous to proposition \ref{p:ef-logic-1} in the cases of (full, unbounded depth) $\comonad{M}$ and $\comonad{P}_2$ \cite{AbramskyDW17,abramsky2021relating}, we obtain a semantic translation of positive $\square$-free modal logic into positive existential two variable logic and graded modal logic into two variable logic with counting quantifiers.
\end{example}

\subsection{Lifting operations to categories of coalgebras}
\label{s:lifting-ops-to-coalgs}
Assuming appropriate equalisers in $\EM{\D}$ exist, we can extend the above lifting to the categories of coalgebras $\lop\colon \prod_i \EM{\C_i} \rightarrow \EM{\D}$. We define $\lop$ pointwise as the equaliser, for coalgebras $(\As_1,\alpha_1)\in \EM{\C_1}$, \dots, $(\As_n,\alpha_n)\in \EM{\C_n}$,
\begin{equation}
    \begin{tikzcd}[column sep=2.5em]
        \lop(\vec{\alpha_i})
            \rar{\iota_{\vec{\alpha_i}}}
        & \EMF{\D}(\op(\vec{A_i}))
            \ar[yshift=0.5em]{rr}{\EMF{\D}(\kappa)\circ \delta}
            \ar[swap,yshift=-0.5em]{rr}{\EMF{\D}(\op(\vec{\alpha_i}))}
        &
        & \EMF{\D}(\op(\vec{\C_i(\As_i)}))
    \end{tikzcd}
    \label{eq:op-lift}
\end{equation}
In fact, the two parallel morphisms are a \df{reflexive pair}. That is, they have a mutual post-inverse.
\begin{restatable}{proposition}{EMLifting}
    \label{p:kl-law}
    If $\kappa$ is a Kleisli-law, and~$\EM{\D}$ has equalizers of reflexive pairs, then $\lop$ is a well-defined functor such that~$\KLF{\D} \circ \op \cong \lop \circ \prod_i \EMF{\C_i}$.
\end{restatable}

Note that existence of said equalisers is guaranteed for our game comonads.

\begin{restatable}{lemma}{EMequalisers}
    \label{l:EM-equalisers}
    For a comonad $\C$ over $\Rel$. If $\C$ preserves embeddings then the category $\EM{\C}$ has equalisers.

    The same is true for a comonad $\C$ over any wide\footnote{Recall that a subcategory is \emph{wide} if it contains the same objects, with only morphisms restricted.} subcategory $\CC$ of $\Rel$ with coproducts, all embeddings of $\Rel$ and satisfying that if $f\circ g\in \CC$, for some $f,g\in \Rel$, then also $g\in \CC$.
\end{restatable}

\subsection{Smooth Kleisli laws}
\label{s:smooth-ops}
In this section we observe that for an operation to preserve $\bisim$ we require the lifted operation to preserve open pathwise-embeddings. Since open pathwise-embeddings are defined in terms of embeddings and paths, we require that the categories $\CC_1,\dots,\CC_{n + 1}$\footnote{It is notationally convenient to denote the codomain category and comonad~$\CC_{n + 1}$ and~$\C_{n + 1}$, rather than~$\CD$ and~$\D$ is this section. As before, we assume that $\EM{\C_{n+1}}$ has equalisers of reflexive pairs.} 
are equipped with proper factorisation systems and the comonads preserve embeddings (cf.\ lemma~\ref{l:lifting-EMfs}).
We then say an operation $\op\colon \prod \CC_i \to \CC_{n + 1}$ with a Kleisli law $\kappa$ is \df{smooth}\footnote{Mirroring the terminology introduced in~\cite{makowsky2004algorithmic}.} if $\lop(f_1,\dots,f_n)$ is an open pathwise-embedding for every tuple of open pathwise-embeddings $f_1,\dots,f_n$. Smoothness immediately guarantees FVM-type theorem for back-and-forth equivalence.

\begin{restatable}{proposition}{SmoothKleisliPreservesEquivalence}
    \label{p:smooth-Kleisli-law}
    If $\op$ admits a smooth Kleisli law, then $\op$ preserves $\bisim$.
\end{restatable}

Recall that $\equiv_{\FO_k}$ is obtained as the $\bisim_{\Ek}$ relation preceded by the $\trI$ translation. The following is one of our main definitions in this paper, tailored to account for this translation.

\begin{definition}
    \label{d:FO-smoothness}
    We say that an `operation' functor $\op\colon \prod_{i=1}^n \R_0(\sg_i) \to \R_0(\sg_{n+1})$ is \df{$(\C_1,\dots,\C_{n+1})$-smooth relative to $\tr_1,\dots,\tr_{n+1}$}, where $\tr_i\colon \R_0(\sg_i) \to \CC_i$ is a functor and $\C_i$ a comonad on $\CC_i$, for $1\leq i \leq n+1$, if there exists a functor $\op^\tr\colon \prod_{i=1}^n \CC_i \to \CC_{n+1}$ and natural transformations
    \[  \theta\colon \op^\tr \circ \prod\nolimits_{i=1}^n \tr_i \to \tr_{n+1} \circ \op
        \qtq{and}
        \kappa\colon \C_{n+1} \circ \op^\tr \to \op^\tr \circ \prod\nolimits_{i=1}^n \C_i
    \]
    such that $\kappa$ is a smooth Kleisli law and $\EMF{\C_{n+1}}(\theta)$ is a collection of open pathwise embeddings. 

    Lastly, we say that $\op$ is \df{FO-$(\C_1,\dots,\C_{n+1})$-smooth} if it is $(\C_1,\dots,\C_{n+1})$-smooth relative to $\tr_1,\dots,\tr_{n+1}$ where each $\tr_i\colon  \R_0(\sg_i) \to \CC_i$ is a codomain restriction of the usual FO translation functor $\trI\colon \R_0(\sg_i) \to \Rs(\sgI_i)$.
\end{definition}

By convention, call an $n$-ary operation $\C$-smooth if it is $(\C,\dots,\C)$-smooth, for the same game comonad $\C$ but possibly $n$ different instances for $n$ different signatures.

\begin{restatable}{corollary}{EkSmoothOpPreserve}
    FO-$\Ek$-smooth operations preserve $\equiv_{\FO_k}$.
\end{restatable}
Similar statements also hold for $\Pk$, $\Mk$ and their corresponding logical equivalences.

\begin{example}
In the case of first-order equivalence, and the coproduct operation, definition \ref{d:FO-smoothness} is instantiated with every comonad $\comonad{C}_i$ as~$\Ek$, and $\tr_i = \trI$. The natural transformation~$\kappa$ is as mentioned in section \ref{sec:motivating} and $\theta$ is the isomorphism $\trI(A) + \trI(B) \cong \trI(A + B)$. Therefore, coproducts are FO-$\Ek$-smooth and preserves $\equiv_{\FO_k}$.
\end{example}

\subsection{Practical axioms for smoothness}
\label{s:sufficient-axioms}
The notion of smoothness in definition~\ref{d:FO-smoothness} is extremely general, but the required conditions can be cumbersome to verify.
We now identify sufficient conditions for a Kleisli-law to be smooth, that are easier to establish in practice. We introduce the following axioms:
\begin{axioms}
    \item[\textbf{\namedlabel{ax:s1}{(S1)}}]
    $\op\colon \prod_i \CC_i \to \CD$ restricts to embeddings.

    \item[\textbf{\namedlabel{ax:s2}{(S2)}}]
    Any path embedding $e\colon \Pc \embed \lop(\vec{\Ac_i})$ has a \df{minimal decomposition} as $e_0\colon \Pc \to \lop(\vec{\Pc_i})$ followed by $\lop(\vec{e_i})\colon \lop(\vec{\Pc_i}) \to \lop(\vec{\Ac_i})$,
        for some path embeddings $e_i\colon \Pc_i \embed \Ac_i$, for $1 \leq i \leq n$.
\end{axioms}

Minimality in \ref{ax:s2} expresses that for any other decomposition of $e$ as $g_0 \colon \Pc \to \lop(\vec{\Qc_i})$ followed by $\lop(\vec{g_i}) \colon \lop(\vec{\Qc_i}) \to \lop(\vec{\Ac_i})$,
for some path embeddings $g_i\colon \Qc_i \embed \Ac_i$, there exist (necessarily unique) morphisms $h_i\colon \Pc_i \to \Qc_i$ such that $e_i = g_i \circ h_i$, for $i=1,\dots,n$. 

A necessary step before showing sufficiency of \ref{ax:s1} and \ref{ax:s2}, is checking the following.

\begin{restatable}{proposition}{PEPreserved}
    If $f_1,\dots,f_n$ are embeddings (or pathwise-embeddings) in $\EM{\C_1}$, \dots, $\EM{\C_n}$, respectively, then $\lop(f_1,\dots,f_n)$ is an embedding (or pathwise-embedding).
    \label{p:pe-preserved}
\end{restatable}

The main difficulty is in verifying that openness for pathwise-embeddings is also preserved, which is the case.

\begin{restatable}{theorem}{SmoothnessTheorem}
    \label{t:smoothness}
    A Kleisli law which satisfies \ref{ax:s1} and \ref{ax:s2} is smooth.
\end{restatable}


\subsection{Example: edge creation}
\label{s:FO-example-edge-creation}

We apply our comonadic analysis to the investigation of bounded clique-width structures in section \ref{s:cw-thm}. 
An operation important in the construction of bounded clique-width structures is edge-creation $\etaRi$, for some $r$-ary $R \in \sigma$, where $\sg$ is without constants.
The operation $\etaRi$, as all the operations for forming a $\sg$-structure of clique-width $\leq m$, apply to structures in an expanded signature\footnote{In fact, this is slightly more general. Clique-width is typically defined for graphs rather than relational structures.} $\sg_m = \sg \cup \{P_1,\dots,P_m\}$ where $P_i$ are unary relations, or `colours'. The operation $\etaRi$ takes a list of colours $i = (i_1,\dots,i_r)$ with $i_j \in \setm$ and adds all tuples in $P_{i_1} \times \dots \times P_{i_r}$ to the interpretation of $R$.

It is clear that $\etaRi \colon \R(\sgm) \to \R(\sgm)$ is a functor. Furthermore, we have a natural transformation
$\kappa\colon \Ek \circ \etaRi \to \etaRi \circ \Ek$
with components $\kappa_{\As}$ that send $w \in \Ek(\etaRi(\As))$ to the same word $w$ in $\etaRi(\Ek \As)$ which satisfies
\ref{ax:kl-law-counit}, \ref{ax:kl-law-comultiplication}, \ref{ax:s1} and \ref{ax:s2}.

\begin{restatable}{lemma}{etaRiSmooth}
    \label{l:etaRi-FO-Ek-smooth}
    $\kappa\colon \Ek \circ \etaRi \to \etaRi \circ \Ek$ is a smooth Kleisli law.
\end{restatable}

\smallskip
For any signature $\sg$, $\etaRi$ does not change the interpretation of any relation in $\sg \backslash \{R\}$, so we obtain a natural isomorphism $\theta:\etaRi(\trI) \cong \trI(\etaRi)$ with components given by the identity as a set function. Therefore, $\etaRi$ is FO-$\Ek$-smooth. 

\subsection{Example: parallel sums}
\label{s:FO-example-coproducts}

Maybe the most prominent example of a FVM theorem is the one for coproducts. Here we verify smoothness for a generalisation of coproducts. Let $\sg$ be a relational signature without constants and let $\sg_1$, $\sg_2$ and $\sg_3$ be extensions of $\sg$ with constants such that $\const{\sg_1}\cup \const{\sg_2} = \const{\sg_3}$. 
This operation is slightly more subtle than the previous examples. In order for our Kleisli law to be natural, we need to restrict the homomorphisms in our categories of interest. Therefore, we work with the restriction $\Rnc(\sg)$ of $\Rel$ to the morphisms which send non-constant elements to non-constant elements. Define a binary operation
\[ \plus\colon \R_0(\sg_1)\times \R_0(\sg_2) \to \R_0(\sg_3) \]
sending $\As\in \R_0(\sg_1)$ and $\Bs\in \R_0(\sg_2)$ to the quotient $(A+B)/_{\approx}$ where $A + B$ is the $\sg$-structure coproduct and $\approx$ the least equivalence relation identifying pairs $(c^\As,c^\Bs)$ such that $c\in \const{\sg_1}\cap \const{\sg_2}$. For a constant $c\in \const{\sg_1}$ we set $c^{\As\plus\Bs} = \inc(c^\As)$ and, similarly, for  $c\in \const{\sg_2}$ set $c^{\As\plus\Bs} = \inc(c^\Bs)$, where $\inc$ is the quotient map $\As + \Bs \to \As \plus \Bs$. Observe that $\inc(\As)\cap \inc(\Bs) = \{ c^{\As\plus\Bs} \mid c\in \const{\sg_1}\cap\const{\sg_2} \}$.

The operation $\plus$ extends to a functor $\Rnc(\sg_1)\times \Rnc(\sg_2) \to \Rnc(\sg_3)$. Furthermore, it admits a smooth Kleisli law w.r.t.\ a comonad $\Eknc$ on $\Rnc(\sg_i)$ (for $i=1,2,$ or $3$), which is the restriction of $\Ek$ to words that do not contain any constant letter. Then, the Kleisli law is the following mapping, for every $\As\in \R(\sg_1)$ and $\Bs\in \R(\sg_2)$, defined on non-constants by
\begin{align*}
    \kappa\colon \Eknc(\As \plus \Bs) \longrightarrow \Eknc(\As) \plus \Eknc(\Bs),
    &
    &
    w \in \wrdk{(A\plus B)}
    &\mapsto
    \begin{cases}
        \ex_1(w) & \text{if } \counit(w) \in \inc(A) \\
        \ex_2(w) & \text{if } \counit(w) \in \inc(B) \\[0.5em]
    \end{cases}
\end{align*}
where $\ex_1([x_1,\dots,x_n])$ is the word\footnote{We use the Haskell-style list comprehension notation, for a word $w = [x_1,\dots,x_n]$ the word $[ x_i \mid \varphi(i)]$ is the restriction of $w$ to the letters on positions $i$ such that $\varphi(i)$ holds.} $[x_i \mid x_i \in \inc(A)]$ and, similarly, $\ex_2([x_1,\dots,x_n])$ is $[x_i \mid x_i \in \inc(B)]$. One can verify that $\kappa$ satisfies axioms \ref{ax:s1} and \ref{ax:s2}, which gives us:

\begin{restatable}{lemma}{KappaCoproducts}
    \label{l:kappa-coproducts}
    $\kappa$ is a smooth Kleisli law $(\Eknc \circ \plus) \rightarrow (\plus \circ (\Eknc \times \Eknc))$.
\end{restatable}

\smallskip
In fact, $\plus$ is FO-$\Eknc$-smooth as $\trI(A) \plus\I \trI(B) \cong \trI(A \plus B)$, where $\plus\I$ is the $\sgI$ variant of $\plus$. Further, it also follows that $\trI(A) \bisim_{\Eknc} \trI(B)$ iff $A \equiv_{\FO_k} B$, giving us the following.

\begin{restatable}{proposition}{FOlogicEquivCoproducts}
    If $A_i \equiv_{\FO_k} B_i$ in $\Rnc(\sg_i)$, for $i=1,2$, then $A_1 \plus A_2 \,\equiv_{\FO_k}\, B_1 \plus B_2$.
\end{restatable}

\subsection{Other examples}
\begin{example}[Products for $\FO$]
\label{e:fo-products}
For any comonad, there is a canonical Kleisli law for products, and when the Eilenberg-Moore category has equalizers, the lifted functor yields products.
By observing that pairs are equal if and only if each component is equal, we obtain an isomorphism $\trI(\As \times \Bs) \cong \trI(\As) \times \trI(\Bs)$.
We can confirm that products of coalgebras preserve open pathwise embeddings,
and therefore products are $\FO$-$\Ek$-smooth.

\end{example}
\begin{example}[Comonad Morphisms]
Example \ref{e:comonad-morphisms} exhibited a comonad morphism $\kappa:\comonad{M} \rightarrow \comonad{P}_2$ yielding a semantic translation of positive $\square$-free modal logic into two-variable positive existential logic and graded modal logic into two-varable logic with counting quantifiers. The intermediate translation from full modal logic into two-variable logic is witnessed by $\kappa$ being $(\comonad{M},\comonad{P}_2)$-smooth. 
\end{example}



\section{Comonads for MSO}

\subsection{\texorpdfstring{$\Ek$}{Ek} for many-sorted structures}
\label{s:Ek-multi-sort}
In order to extend our use of $\Ek$ to $\MSO$ we interpret the syntax of $\MSO$ as a two-sorted first-order logic. 
We then use a translation functor to fix the semantics of $\MSO$ to coincide with the standard semantics.
This requires us to `upgrade' the $\Ek$ comonad to handle two-sorted structures. For simplicity and generality, we extend our definition of $\Ek$ to accommodate all many-sorted structures for fixed set $S$ of sorts. 

A signature $\sg$ on sorts $S$ contains a set of constant symbols $\consts{\sg}{s}$ for each $s \in S$ and a set of relation symbols $R$ with an associated arity given by a tuple of sorts ${\langle s_1,\dots,s_r \rangle \in S^r}$.
An object $\As$ in the category of $\sg$-structures $\RelS$ has a universe given by an $S$-indexed family of sets $\{A_{s}\}_{s \in S}$ along with interpretations $c^{\As} \in A_s$ for every constant symbol $c \in \consts{\sg}{s}$ and $R^{\As} \sue A_{s_1} \times \dots \times A_{s_r}$ for every relation symbol with arity $\langle s_1,\dots,s_r \rangle$ in $\sg$.
A morphism $f:\As \rightarrow \Bs$ in $\RelS$ is a $S$-indexed family of functions $\{f_s:A_s \rightarrow B_s\}$ which preserves all constants and all relations.

These objects provide the semantics for $S$-sorted first-order logic.
Both syntax and semantics of $S$-sorted first order logic $\FO(S)$ are similar to the single-sorted case. The key difference is that terms (constants and variables) always have an associated sort. 
Consequently, the quantifiers $\exists x{:}s$, $\forall x{:}s$ range over the set $A_{s}$ when interpreted in the model $\As$. 
For subset $J \subseteq S$, $\FO_k(S,J)$ denotes $S$-sorted first order logic up to quantifier rank $k$ with equality for terms with sorts $s \in J$. 

We can now describe $\Ek$ as a comonad over $\RelS$. 
The universe of $\Ek \As$ for object $\As$ in $\RelS$ has components $(\Ek \As)_{s}$ equal to the disjoint union of $\consts{\sg}{s}$ and the set of words of the form $[a_1{:}s_1,\dots,a_n{:}s_n]$ where $n \leq k$ and, for all $i \leq n$, $s_i \in S$, $a_i \in A_{s_i}$ and $s_n = s$. 
The counit $\varepsilon_{\As}$ for $A \in \RelS$ and coextension $f^{*}:\Ek\As \rightarrow \Ek\Bs$ for $f:\Ek\As \rightarrow \Bs$ are defined as in the single-sorted case while respecting the sorts.  
We can perform a similar upgrade to the single-sorted $\Eknc$ comonad to obtain a many-sorted $\Eknc$ comonad. 
\begin{restatable}{proposition}{ManySortedEkIsComonad}
\label{p:many-sorted-ek-is-comonad}
$\Ek$ and $\Eknc$ are comonads over $\RelS$ and $\Rsnc(\sg)$, respectively.
\end{restatable}
Just as with the single-sorted case, to capture equality for the sorts $J \subseteq S$, we consider an extended signature $\sgJ$ which has additional relations $I_{s}$ of arity $\langle s,s \rangle$ for every $s \in J$. Let $\trJ\colon\R^S_0(\sg) \rightarrow\RelJ$ be the functor that sends $\As$ to the same structure where the additional relations are interpreted as $(a,b) \in I^{\As}_{s}$ iff $a = b \in A_s$ for $s \in J$. 
Let $\trJnc\colon\R^S_0(\sg) \rightarrow \Rsnc(\sgJ)$ be the codomain restriction of $\trJ$.  
We then show that a version of proposition~\ref{p:ef-logic-2} still holds for these many-sorted $\Ek$ and $\Eknc$ comonads. 

\begin{restatable}{proposition}{ManySortedEkbisim}
\label{prop:many-sorted-ek-bisim}
Given $\As, \Bs \in \Rs_0(\sg)$,
\begin{itemize}
    \item $\As \equiv_{\FO_k(S,J)} \Bs$  \ee{iff} $\trJ \As \bisim_{\Ek} \trJ \Bs$. 
    \item $\As \equiv_{\FO_k(S,J)} \Bs$  \ee{iff} $\trJ_{\text{nc}} \As \bisim_{\Eknc} \trJ_{\text{nc}} \Bs$, if $s\in J$ for every sort $s\in S$ such that $\consts{\sg}{s}\not= \emptyset$.
\end{itemize}
\end{restatable}


\subsection{The two-sorted translation for \texorpdfstring{$\MSO$}{MSO1}}
\label{s:two-sorted-transl}
Second-order logic can be treated syntactically as a many-sorted first-order logic (see e.g. \cite{vaananen2021Stanford,manzanoBook}). 
In particular, we consider $\MSO_k$ on signature $\sg$ as $\FO_k(S,J)$ with $S = \{\fst,\snd\}$ and $J = \{\fst\}$ on an expanded signature $\sgE$ where $\fst$ is sort for ordinary {\underline f}irst-order terms, and $\snd$ is a sort for {\underline m}onadic second-order set variables. The expanded $S$-sorted signature $\sgE$ has constants $\consts{\sgE}{\fst} = \const{\sg}$ and $\consts{\sgE}{\snd} = \varnothing$; and relations $R$ of arity $\langle \fst,\dots,\fst \rangle$ ($\fst$ repeated $r$ times) for every $R \in \sg$ of arity $r$, and also $e$ of arity $\langle \fst,\snd \rangle$ and $I$ of arity $\langle \fst,\fst \rangle$.

The additional expressive power of $\MSO$ over $\FO(S,J)$ is obtained by restricting the semantics of $\MSO$ to only include the class of  `standard' models of $\MSO$ rather than all $\sgE$-structures.
We pick out these standard models using a functor $\trE\colon\R_0(\sg) \rightarrow \RelE$ where, for a $\sg$-structure $\As$, the two sorts are instantiated as $\trE(\As)_\fst = A$ and $\trE(\As)_\snd = \Pw(A)$, the constants $c^{\trE\As} = c^{\As}$ and relations
\[ R^{\trE(A)} = R^A, \text{ for } R\in \sg, \qquad e^{\trE(\As)}(a,M) \ete{iff} a\in M \qquad I^{\trE(\As)}(a,a') \ete{iff} a = a'\] 
Let $\trEnc\colon\R^S_0(\sgE) \rightarrow \Rsnc(\sgE)$ be the codomain restriction of $\trE$.
We can adapt the syntactic translation stated in \cite{vaananen2021Stanford} to translate $\MSO_k$ sentences in signature $\sg$ (i.e.\ $\MSO$ sentences of quantifier rank ${\leq}\,k$) to $\FO_k(S,J)$ sentences in $\sgE$.
Linking this syntactic translation to the semantic translation~$\trE$ and utilizing proposition \ref{prop:many-sorted-ek-bisim}, we are able to capture $\equiv_{\MSO_k}$ using the many-sorted $\Ek$ and~$\Eknc$:

\begin{restatable}{theorem}{MSOManySortedEk}
\label{thm:mso-manysortedEk}
    $\As \equiv_{\MSO_k} \Bs$ \ee{iff} $\trE \As \bisim_{\Ek} \trE \Bs$ \ee{iff} $\trE \As \bisim_{\Eknc} \trE \Bs$
\end{restatable}

The last equivalence follows from the second item of proposition~\ref{prop:many-sorted-ek-bisim} as the condition is satisfied by $J = \{\fst\}$ and $\consts{\sgE}{\fst} = \const{\sg}$. Next we also adapt definition~\ref{d:FO-smoothness} to this setting.

\begin{definition}
\label{d:MSO-smoothness}
We say that $\op$ is \df{$\MSO$-$(\C_1,\dots,\C_{n+1})$-smooth} if it is $(\C_1,\dots,\C_{n+1})$-smooth relative to $\tr_1,\dots,\tr_{n+1}$ where each $\tr_i\colon  \R_0(\sg_i) \to \CC_i$ is a codomain restriction of the MSO translation functor $\trE\colon \R_0(\sg_i) \to \Rs(\sgE_i)$.
\end{definition}

As before, proposition~\ref{p:smooth-Kleisli-law} and theorem~\ref{thm:mso-manysortedEk} imply that the MSO-$\Ek$-smooth (or MSO-$\Eknc$-smooth) functors preserve $\equiv_{\MSO_k}$.

\subsection{Example: edge creation}
\label{s:MSO-example-edge-creation}
In this section, we demonstrate the $\etaRi$ operation is $\MSO$-$\Ek$-smooth. Since this operation is used in generating structures of bounded clique-width, this is a step towards recovering Courcelle's theorem for bounded clique-width structures using our comonadic framework. 

By definition the functor $\etaRi\colon\Relk \rightarrow \Relk$ only changes the interpretation of $R \in \sgm$. 
In the two-sorted setting of $\sgE$, we define a variant of this functor $\etaRiE$ of type $\RelEm \rightarrow \RelEm$ which only changes interpretation of the corresponding $R \in \sgEm$ with arity $\langle \fst,\dots,\fst \rangle$. This ensures that there is an isomorphism $\theta:\etaRiE(\trE(\As)) \cong \trE(\etaRi(\As))$, for any $\sgm$-structure $\As$. In addition to the existence of $\theta$, in order for $\etaRi$ to be $\MSO$-$\Ek$-smooth, we need a smooth Kleisli law
\[ \kappa\colon\Ek \circ \etaRiE \rightarrow \etaRiE \circ \Ek. \]
The components of $\kappa_{\As}$ are given by $S$-sorted functions where $\kappa_{\As,t}$ sends a word $w \in (\Ek(\etaRiE(\As)))_{t}$ to the same word $w \in (\etaRiE(\Ek(\As)))_{t}$ for $t \in S = \{\fst,\snd\}$.
The proof that this $\kappa$ satisfies \ref{ax:s1} and \ref{ax:s2} is a direct adaptation of the single-sorted proof (cf.\ section~\ref{s:FO-example-edge-creation}).

\begin{restatable}{lemma}{EdgeCreationMSOSmoothKleisli}
    \label{l:etaRi-MSO-Ek-smooth}
    For every $r$-ary relation $R \in \sg$ and $r$-tuple of labels $\vec{i}$, the operation $\etaRi$ is $\MSO$-$\Ek$-smooth.
\end{restatable}

\subsection{Example: twisted sums}
\label{s:twisted-sums}

Here we show that $\plus\colon \R_0(\sg_1)\times \R_0(\sg_2) \to \R_0(\sg_3)$, defined in section~\ref{s:FO-example-coproducts}, is $\MSO$-$\Eknc$-smooth.
To this end, we define a ``twisted sum'' functor
\[ \bowtie\colon \Rsnc(\sg_1\E)\times \Rsnc(\sg_2\E) \to \Rsnc(\sg_3\E) \]
as the lift of $\plus$ along $\trE$. For $\As \in \Rsnc(\sg_1\E)$ and $\Bs\in \Rsnc(\sg_2\E)$, the $\sgI_3$-structure reduct of $\As \bowtie \Bs$ is precisely $\As' \plus \Bs'$ where $\As'$ and $\Bs'$ are the $\sgI_1$-- and $\sgI_2$-reducts of $\As$ and $\Bs$, respectively. As before, we have the map $\inc\colon A_\fst + B_\fst \to (A_\fst + B_\fst)/_{\approx}$. What remains to define is the second sort $(\As \bowtie \Bs)_\snd$, which is just $\As_\snd \times \Bs_\snd$, and the relation $e^{\As \bowtie \Bs}$, which is the set $\{ (\inc(a), (a',b')) \mid e^\As(a,a'),\, b'\in \Bs_\snd \}\cup \{ (\inc(b), (a',b')) \mid a'\in \As_\snd,\, e^\Bs(b,b') \}$.

We now define the natural transformation $\theta\colon \trE(\As) \bowtie \trE(\Bs) \to \trE(\As \plus \Bs)$. In the $\fst$-sort it is just the identity function and in the $\snd$-sort we send the pair $(M,N)$ of subsets $M\sue A$ and $N \sue B$ to the subset $\inc(M)\cup \inc(N)$ of $\As \plus \Bs$. Clearly, $\theta$ is a strong onto homomorphism in $\Rsnc(\sgE_3)$ and, therefore, $F^{\Eknc}(\theta)$ is an open pathwise embedding.

Lastly, we need a smooth Kleisli law $\kappa\colon \Eknc(\As \bowtie \Bs) \to \Eknc(\As) \bowtie \Eknc(\Bs)$. Analogously to section~\ref{s:FO-example-coproducts}, define $\ex_1$ by sending a word $w = [x_1:s_1, \dots, x_n:s_n]$ from $\Eknc(\As \bowtie \Bs)$ to $\ex_1(w) = [ y_i : s_i \mid x_i \in \inc(\As_\fst) \text{ or } s_i = \snd]$, where $y_i = a$ if $x_i = \inc(a)$ or $y_i = a'$ if $s_i = \snd$ and $x_i = (a',b')$. Then, $\ex_2$ is defined dually. We define $\kappa(w)$, for a non-constant $w$, as $\ex_1(w)$ if $\counit(w) \in \inc(\As_\fst)$, $\ex_2(w)$ if $\counit(w) \in \inc(\Bs_\fst)$ and as $(\ex_1(w),\ex_2(w))$ for $w$ from the $\snd$-sort.

\begin{restatable}{proposition}{plusMSOsmooth}
    \label{p:oplus-MSO-Ek-smooth}
    $\kappa$ is a smooth Kleisli law and $\plus$ is $\MSO$-$\Eknc$-smooth.
\end{restatable}

\subsection{Extensions of \texorpdfstring{$\MSO$}{MSO}}
The generality of definition \ref{d:FO-smoothness} and the many-sorted $\Ek$ allows us to adapt our approach to proving FVM and consequently Courcelle theorems for fragments and extensions of (monadic) second order logic. We can accomplish this by choosing the appropriate translation functors~$\tr$. For example, the logic $\MSO 2$, which allows quantification over subsets of tuples in relations, could be captured by having a sort $\mathbf{s}_R$ for each relation $R \in \sg$ and the translation functor $\tr$ on a structure $\As$ sets $\tr(\As)_{\mathbf{s}_R}$ to $\Pw(R^{\As})$. Similar translations can be tailored for weak MSO, counting MSO, and $n$-adic second-order logic.

\begin{remark}[Failure of FVM theorems]
Unlike the case of $\FO$ in example \ref{e:fo-products}, there famously fails to be a FVM product theorem for $\MSO$ (see e.g. \cite{makowsky2004algorithmic}). In our framework, the failure of $\times$ to be $\MSO$-$\Ek$-smooth is demonstrated by the non-existence of a natural transformation $\theta:\trE(\As) \times \trE(\Bs) \rightarrow \trE(\As \times \Bs)$ where the components of $\EMF{\Ek}(\theta)$ are open pathwise embeddings. Similarly, the non-existence of a suitable $\theta$ or smooth Kleisli law $\kappa$ could be used to demonstrate the impossibility of a FVM product theorem for $\MSO 2$ and a FVM coproduct theorem for $n$-adic second-order logic with $n \geq 2$. 
\end{remark}

\section{Comonadic Courcelle's theorems}
\label{s:courcelle}

In the following we state an abstract version of Courcelle's theorem, for classes of $\tau$-structures given by a \df{$\tau$-presentation $\left< \Gens, \Ops\right>$ over $\Sigma$}. This is specified by
\begin{itemize}
    \item a finite set of signatures $\Sigma$, which are extensions of $\tau$,
    \item a finite set of generators $\Gens$, where each $B \in \Gens$ is an object of $\R_0(\sg)$, for some $\sg\in \Sigma$, and
    \item a finite set of operations $\Ops$, where each $H\in \Sigma$ is a functor $\prod_{i=1}^n \R_0(\sg_i) \to \R_0(\sg)$, with $\sg\in \Sigma$ and $\sg_i \in \Sigma$, for every $i$.
\end{itemize}
Intuitively, $\left< \Gens, \Ops\right>$ is a presentation of the class $\Gamma \sue \R_0(\tau)$ where $A \in \Gamma$ if and only if there is a term $t_A$ using objects $B\in \Gens$ as constants and functors $\op\in \Ops$ as operations, such that the $\tau$-structure reduct of $t_A$ is isomorphic to $A$.
Examples of classes with such presentations are the classes of structures of tree-width ${\leq}\,k$ or clique-width ${\leq}\,k$. These are the focus of our attention in the following sections, once we state our comonadic Courcelle theorem.

To this end, we say that a comonad $\C$ over $\CB$ has \df{finite $\tr$-type}, for a functor $\tr\colon \CA \to \CB$, if the equivalence relation $\bisim_\C^\tr$ has finitely many equivalence classes, where $A \bisim_\C^\tr B$ iff $\tr A \bisim_\C \tr B$. 
A standard tree-automata argument yields the following.

\begin{restatable}{theorem}{ComonadicCourcelle}
    \label{t:comonadic-courcelle}
    Let $\Delta$ be a class of $\tau$-structures and let $\left< \Gens, \Ops\right>$ be a $\tau$-presentation over $\sg_1$, \dots, $\sg_n$ and, for $1 \leq i \leq n$ let $\tr_i\colon \R_0(\sg_i) \to \CC_i$ be a functor and $\C_i$ a comonad over~$\CC_i$. Assume further that
    \begin{enumerate}
        \item $\CC_i$ is a category with a proper factorisation system,
        \item $\C_i$ has finite $\tr_i$-type, preserves embeddings and $\EM{\C_i}$ has equalisers of reflexive pairs,
        \item operations in $\Ops$ are $(\vec{\C_i})$-smooth relative to $\vec{\tr_i}$, and
        \item for every $1 \leq i \leq n$, the class of $\sg_i$-structures with their $\tau$-structure reduct in $\Delta$ is closed under $\bisim_{\C_i}^{\tr_i}$.
    \end{enumerate}
    Then there exists an algorithm, which for a given $\tau$-structure $A$ given as a term $t_A$ of $\left< \Gens, \Ops \right>$, decides $A \in \Delta$ in linear time, in the size of $t_A$.
\end{restatable}

This theorem directly applies to our setting of $\MSO$-$\Eknc$-smooth operations. Indeed, the (surjective, strong injective) factorisation system of $\Rs(\sgE)$ restricts to $\Rsnc(\sgE)$.
Further, observe that $\Eknc$ preserves embeddings, has finite $\trEnc$-type, because there are only finitely many $\equiv_{\MSO_k}$-equivalence classes (cf.\ proposition~\ref{prop:many-sorted-ek-bisim}), and the category $\EM{\Eknc}$ has equalisers (cf.\ lemma~\ref{l:EM-equalisers}).
An $\MSO$ sentence $\varphi$ of quantifier rank ${\leq}\,k$, in signature $\tau$, defines the class $\Delta_\varphi$ of $\tau$-structures such that $A\in \Delta_\varphi$ iff $A \models \alpha$. Observe that $\Delta_\varphi$ automatically satisfies item 4 in theorem~\ref{t:comonadic-courcelle}.
Indeed, $\varphi$ is a valid $\MSO$ sentence in any signature $\sg$ extending $\tau$. Consequently, if $A \bisim_{\Ek}\E B$ and the $\tau$-structure reduct of $A$ is in $\Delta_\varphi$ then so is $B$.
Therefore:

\begin{corollary}
    \label{c:abstract-MSO-Ek-courcelle}
    Let $\left< \Gens, \Ops\right>$ be a $\tau$-presentation and $\varphi$ be a $\MSO$ sentence in signature $\tau$ of quantifier rank ${\leq}\,k$.

    If every $\op \in \Ops$ is $\MSO$-$\Eknc$-smooth then there is an algorithm deciding $A \models \varphi$, for an arbitrary $A$ given as a term $t_A$ of $\left< \Gens, \Ops\right>$, running in linear time, in the size of $t_A$.
\end{corollary}

\subsection{Courcelle's tree-width theorem}
\label{s:tw-thm}

Tree-width is an important, well-understood graph parameter, popularised by the work of Robertson--Seymour on graph minors \cite{robertson1986graph}. It is well-known~\cite{courcelle1993graph,courcelle2012graph} that a graph or, more generally a $\sg$-structure, has tree-width ${<}\, m$ if and only if it belongs to the class presented by $\left< \Gens_{tw},\Ops_{tw} \right>$ over $\Sigma_{tw}$ where
\begin{itemize}
    \item $\Sigma_{tw} = \{ \sg_C \mid C \sue \{c_1,\dots,c_m\}\}$ where $\sg_C$ the extension of $\sg$ by constants in $C$,
    \item $\Gens_{tw}$ consists of all $\sg_C$-structures on ${\leq}\,m$ vertices, with $C \sue \{c_1,\dots,c_m\}$, and
    \item $\Ops_{tw}$ consists of all parallel sums $\plus$ of type
    $\R(\sg_C) \times \R(\sg_D) \to \R(\sg_{C \cup D})$, for every $C,D$, and operations $\forgCD\colon \R(\sg_C) \to \R(\sg_D)$, for $D \sue C$, where $\forgCD(A)$ is the $\sg_D$-structure reduct of the $\sg_C$-structure $A$.
\end{itemize}

In proposition~\ref{p:oplus-MSO-Ek-smooth} we have established that the parallel sum operations in $\Ops$ are $\MSO$-$\Eknc$-smooth. Furthermore, checking smoothness for $\forgCD$ is also straightforward.
Consequently, corollary~\ref{c:abstract-MSO-Ek-courcelle} for $\left< \Gens_{tw},\Ops_{tw} \right>$ entails the classical Courcelle's theorem.
\begin{theorem}
    \label{t:classical-tw-recovered}
    Let $\varphi$ be a $\MSO$ sentence in signature $\tau$ of quantifier rank ${\leq}\,k$. Then, there exists an algorithm deciding $A \models \varphi$, for an arbitrary $A$ given by a term $t_A$ of $\left< \Gens_{tw}, \Ops_{tw}\right>$, running in linear time, in the size of $t_A$.
\end{theorem}

\begin{remark}
    In traditional presentations of Courcelle's theorem, the algorithm first computes the term $t_\As$, given a structure $\As$ of tree-width $\leq k$. This part of the algorithm also runs in linear time in the size of $\As$ \cite{bodlaender1997treewidth}.
\end{remark}



\subsection{Courcelle's clique-width theorem}
\label{s:cw-thm}
In order to recover Courcelle's theorem for structures of bounded clique-width as a corollary of theorem~\ref{t:comonadic-courcelle}, we note that the class of structures of clique-width $\leq m$ is given by the $\sg$-presentation $\left< \Gens_{cw},\Ops_{cw} \right>$ over $\Sigma_{cw}$ where
\begin{itemize}
    \item $\Sigma_{cw} = \{\sgm\}$, recall that $\sgm = \sg \cup \{P_1,\dots,P_m\}$ with each $P_i(\cdot)$ unary
    \item $\Gens_{cw}$ is the set of $\sgm$-strutures $\{v_{1},\dots,v_{m}\}$ where $v_{i}$ is the singleton $\{*\}$ such that $P^{v_i}_i = \{*\}$ and $R^{v_{i}} = \varnothing$ for all $R \in \sgm \setminus \{P_i\}$ 
    \item $\Ops_{cw}$ consists of coproducts $\oplus$ of type $\R_0(\sgm)\times \R_0(\sgm) \to \R_0(\sgm)$,  edge-creation operations $\etaRi$ for all $R \in \sg$ of arity $r$ with $\vec{i} \in \setm^r$, and relabeling operations $\rhoij$ for all $i,j \in \setm$. The relabeling operation $\rhoij$ `recolours' elements of colour $i$ with the colour $j$: 
 \[ P^{\rhoij(\As)}_j = P^{\As}_i \cup P^{\As}_j \qquad 
            P^{\rhoij(\As)}_i = \varnothing \qquad 
            R^{\rhoij(\As)} = R^{\As} \text{, for } R \in \sgm \setminus \{P_i,P_j\} \]
\end{itemize}
The coproduct $\oplus$ of type $\R_0(\sgm) \times \R_0(\sgm) \to \R_0(\sgm)$ is a special case of the $\MSO$-$\Ek$-smooth operation $\plus$ mentioned in proposition~\ref{p:oplus-MSO-Ek-smooth}. In lemma~\ref{l:etaRi-MSO-Ek-smooth} we established all the edge-creation operations $\etaRi$ are $\MSO$-$\Ek$-smooth. Verifying that the relabeling operations $\rhoij$ are $\MSO$-$\Ek$-smooth is straightforward.  
Consequently, corollary~\ref{c:abstract-MSO-Ek-courcelle} for  $\left< \Gens_{cw},\Ops_{cw} \right>$ entails Courcelle's theorem for bounded clique-width structures. 

\begin{theorem}
    \label{t:cw-courcelle-recovered}
    Let $\varphi$ be a $\MSO$ sentence in signature $\tau$ of quantifier rank ${\leq}\,k$. Then, there exists an algorithm deciding $A \models \varphi$, for an arbitrary $A$ given by a term $t_A$ of $\left< \Gens_{cw}, \Ops_{cw}\right>$, running in linear time, in the size of $t_A$.
\end{theorem}

\section{Conclusion}
We presented a general semantic framework for FVM-type theorems, based on distributive laws. By varying the choice of game comonads, our results are parametric in the logic of interest, and yield FVM-type results for the logic. Moreover, we proved several new FVM-type results for positive existential and counting quantifier variants of these logics by applying theorem~\ref{t:FVM-from-kleisli-laws} to the relevant game comonad. We combined this with a game comonadic
semantics for monadic second order model equivalence. Both were illustrated by many examples, including FVM-type results spanning these different logics, and detailed expositions of the key operations for constructing graphs of bounded tree-width and clique-width. Combining our accounts of FVM-type theorems and monadic second order logic, we presented an abstract Courcelle theorem. This theorem was then instantiated to recover concrete tree-width and clique-width Courcelle theorems.

For reasons of exposition we have concentrated on relational structures. In fact theorem~\ref{t:comonadic-courcelle}, and other results, can be stated much more generally. It would be interesting to explore algorithmic meta-theorems in such broader settings, for example~\cite{BurtonD17}.

The original Feferman-Vaught theorem~\cite{feferman1967first} deals with operations on indexed families of structures, with the indexing itself carrying structure. We have not needed this level of generality in the present work, but the extension to account for this wider class of operations would be a natural next step.

\interlinepenalty=10000 
\bibliography{fvmc}

\newpage
\appendix
\counterwithin{theorem}{section}
\counterwithin{proposition}{section}
\counterwithin{lemma}{section}
\counterwithin{definition}{section}

\section{Omitted Proofs}

In the following we make use of the following elementary properties of proper factorisations systems.
\begin{lemma}[e.g.\ Section 2 in \cite{freyd1972categories}]
    \label{l:fs-basics}
    Given a proper factorisation system $(\Ec,\Mc)$,
    \begin{enumerate}
        \item $\Ec\cap \Mc = \mathrm{Iso}$,
        \item if $g\circ f \in \Ec$ then $g\in \Ec$ and, dually,\\ if $g\circ f\in \Mc$ then $f\in \Mc$,
        \item $\Ec$ contains all coequalisers and $\Mc$ all equalisers.
    \end{enumerate}
\end{lemma}


We confirm the following generalization to signatures with constants of results in~\cite{abramsky2021relating}.
\StandardComonadsWithConstants*
\begin{proof}
\newcommand{\PN}{\comonad{P}_{\nats}}
\newcommand{\PmN}{\comonad{P}_{m + \nats}}
    In the following $\nats$ denotes the set of natural numbers, \emph{not} a comonad.
    Fix a signature~$\sg$, potentially containing constant symbols.
    In~\cite{AbramskyDW17, abramsky2021relating} a pebbling comonad~$\Pk$ for signatures without constants is defined. In the proofs, the bound on the number of pebbles plays no essential role, so we
    can equally define a comonad $\PN$ on $\R(\rel{\sg})$, with a pebble for each natural number. 
    
    The universe of~$\PN(\As)$ is~$(\nats \times \As)^+$ with $p \in \nats$. The counit extracts the second element of the tail of its input list.
    For~$R \in \rel{\sg}$ of arity~$n$, $R^{\PN(\As)}(s_1,\ldots,s_n)$ if:
    \begin{enumerate}
        \item The~$s_i$ are pairwise comparable in the prefix order.
        \item If~$s_i < s_j$, then the pebble index of the tail of~$s_i$ does not appear in the suffix of~$s_i$ in~$s_j$.
        \item $R^{\As}(\counit(s_1),\ldots,\counit(s_n))$
    \end{enumerate}
    The coextension of~$h : \PN(\As) \rightarrow \Bs$ is given by:
    \begin{align*}
    h^*(&[(p_0, a_0),\ldots,(p_n, a_n)]) := \\
    &[(p_0, [(p_0,a_0)]),\ldots, (p_n, [(p_0,a_0),\ldots,(p_n, a_n)])]
    \end{align*}
    
    Assuming~$\const{\sg}$ has $m$ elements, with a specified linear ordering $c_0,\ldots,c_{m - 1}$, we can restrict~$\PN$ to a comonad~$\PmN$, on~$\R(\rel{\sg})$, by restricting to lists:
    \begin{enumerate}
        \item With length greater than~$m$.
        \item With the $m$ element prefix of each list being \[ [(0,c^\As_0),\ldots,(m - 1, c^\As_{m - 1})] \]
        \item The pebble indices in positions~$m$ and above are greater than or equal to~$m$.
    \end{enumerate}
    We then observe that this yields a comonad on~$\Rel$ if we define for~$0 \leq i < m$:
    \[ c^{\PmN(\As)_i} := [(0,c_0^\As,\ldots,(i,c_i^\As)] \]
    It is easy to confirm the counit and coextensions preserve the constant elements.
    
    We then define:
    \begin{itemize}
        \item A comonad~$\comonad{E}$ as the restriction of~$\PmN$ to lists of the form:
        \[ [(0,a_0),(1,a_1),\ldots,(n,a_n)] \]
        \item A comonad~$\comonad{M}$ as the restriction of~$\comonad{P}_{\nats}$ to lists of the form:
        \[ [(0,a_0),(1,a_1),(0,a_2),\ldots,(\operatorname{parity}(n),a_n)] \]
        such that~$E^{\As}(a_i,a_{i + 1})$ and~$a_0 = c^\As_0$.
        \item A comonad~$\comonad{P}$ as being~$\PmN$.
    \end{itemize}
    Then, up to isomorphism as comonads:
    \begin{itemize}
        \item The comonad~$\Ek$ is given by restricting~$\comonad{E}$ to lists of at most~$m + k$ elements.
        \item The comonad~$\Mk$ is given by restricting~$\comonad{M}$ to lists of at most~$1 + k$ elements.
        \item The comonad~$\Pk$ is given by restricting~$\comonad{P}$ to lists with pebble indices at most~$m + k$.
    \end{itemize}
    In each case, it is routine to verify that coextensions preserve the restricted class of lists.
\end{proof}

The following is a standard result about lifting factorization systems to Eilenberg-Moore categories.
\EMFactorization*
\begin{proof}[Proof sketch.]
    It is clear that both $\ol\Ec$ and $\ol\Mc$ contain all isomorphisms and are closed under compositions. The factorisation of a coalgebra morphism $h\colon (A,\alpha) \to (B,\beta)$ is computed as the $(\Ec, \Mc)$-factorisation of the underlying morphism into a quotient $e$ and an embedding $m$ as shown below.
    \[
        \begin{tikzcd}
            A \rar[->>]{e} \dar{\alpha} & C \rar[>->]{m}\dar[dashed]{\gamma} & B \dar{\beta} \\
            \C(A) \rar{\C(e)} & \C(C) \rar[>->]{\C(m)} & \C(B)
        \end{tikzcd}
    \]
    Since $\C$ restricts to embeddings, $\C(m)$ is an embedding and so there exists a diagonal filler $\gamma$. It is immediate to check that $\gamma$ is a $\C$-coalgebra.

    Next, assume we have a commutative square of coalgebra morphisms, with $e\in \ol \Ec$ and $m\in \ol\Mc$.
    \[
        \begin{tikzcd}
            (A,\alpha) \rar[->>]{e}\dar[swap]{u} & (B,\beta) \dar{v} \ar[swap,dashed]{ld}{d} \\
            (A',\alpha') \rar[>->]{m} & (B',\beta')
        \end{tikzcd}
    \]
    By orthogonality of $e$ and $m$, there exists a unique diagonal filler $d$. It is immediate to show that $d$ is a coalgebra morphism $(B,\beta) \to (A',\alpha')$ from the fact that $e$ is orthogonal to $\C(m)$.

    The ``furthermore'' part follows from the definitions.
\end{proof}

\subsection{Proofs for section~\ref{sec:motivating}}

\InterleavingSumsPreserveOPE*
\begin{proof}
    Recall from section 9 in \cite{abramsky2021relating} that $\EM\Ek$ can be identified with the category of $\sg$-structures equipped with forest order $\sqsubseteq$ of depth ${\leq}\, k$ such that, for any relation symbol $R\in \sg$, $R^\As(x_1,\dots,x_n)$ implies $x_i \sqsubseteq x_j$ or $x_j \sqsubseteq x_i$ \ ($\forall i, j \in \{1,\dots,n\}$). The morphisms in $\EM\Ek$ are then the morphisms of $\sg$-structures which preserve the covering relation $\prec$, where $x \prec y$ iff $x \sqsubseteq y$ and, for every $z$ such that $x \sqsubseteq z \sqsubseteq y$, either $z = x$ or $z = y$.

    It then follows from the definition of $\blift+$ (in terms of the equaliser) that $(\As,\sqsubseteq^\As) \blift+ (\Bs,\sqsubseteq^\Bs)$ consists of precisely the words $w\in \Ek(\As+\Bs)$ such that
    \begin{itemize}
        \item the restriction of $w$ to the sub-word consisting of only letters from $A$, is a path in $(\As,\sqsubseteq^\As)$ and
        \item similarly, the sub-word of $w$ with letters from $B$, is a path in $(\Bs,\sqsubseteq^\Bs)$
    \end{itemize}
    From this it is easy to see that, given open pathwise embeddings $f$ and $g$, also $f \blift+ g$ is an open pathwise-embedding, because checking openness and being a pathwise embeddings for $f \blift+ g$ is decomposed into the corresponding properties of $f$ and $g$.
\end{proof}

\subsection{Proofs for section~\ref{s:fvm-thm-from-klei-laws}}
We shall write~$g \klcomp f$ for composition in Kleisli categories.
The following result is well known if various forms, although full proofs are hard to find in the literature. We explicitly confirm the details for the multi-argument version that we require.
\KleisliCorrespondence*
\begin{proof}
\begingroup
\allowdisplaybreaks
Let~$\kappa : \D \circ \op \rightarrow \op \circ \prod_i \C_i$ be a Kleisli-law. For family of morphisms~$f_i : \C_i \As_i \rightarrow \Bs_i$, define~$\klop(\vec{f_i})$ as the composite:
\[
\begin{tikzcd}
    \D(H(\vec{A_i})) \rar{\kappa_{\As}} & \op(\vec{\C_i A_i}) \rar{\op(\vec{f_i})} & \op(\vec{B_i})
\end{tikzcd}
\]
That identities a preserved is immediate from~\eqref{ax:kl-law-counit}. For composition:
\begin{align*}
    \klop(\vec{g_i \klcomp f_i}) &= \klop(\vec{g_i \circ \C_i(f_i) \circ \delta})\\
    &= \op(\vec{g_i \circ \C_i(f_i) \circ \delta}) \circ \kappa\\
    &= \op(\vec{g_i}) \circ \op(\vec{\C_i(f_i)}) \circ \op(\delta) \circ \kappa\\
    &= \op(\vec{g_i}) \circ \op(\vec{\C_i(f_i)}) \circ \kappa \D(\kappa) \circ \kappa\\
    &= \op(\vec{g_i}) \circ \kappa \circ \D(\op(\vec{f_i})) \circ \D(\kappa) \circ \kappa\\
    &= \klop(\vec{g_i}) \klcomp \klop(\vec{f_i})
\end{align*}
The key step uses~\eqref{ax:kl-law-comultiplication}. To confirm~\eqref{eq:kleisli-lift}:
\begin{align*}
    \klop(\vec{\KLF{\C_i}(f_i)}) &= \klop(\vec{\KLF{\C_i}(f_i)})\\
    &= \klop(\vec{f_i \circ \counit})\\
    &= \op(\vec{f_i \circ \counit}) \circ \kappa\\
    &= \op(\vec{f_i}) \circ \op(\vec{\counit}) \circ \kappa\\
    &= \op(\vec{f_i}) \circ \counit\\
    &= \KLF{\D}(\op(\vec{f_i}))
\end{align*}

Now assume we have~$\op$ and~$\klop$ satisfying~\eqref{eq:kleisli-lift}. 
Note that for~$\As_i$ a~$\CC_i$-object, $\id_{\C_i \As_i}$ is a morphism in~$\Klei{\C_i}$. Define:
\begin{equation*}
    \kappa_{\As} := \klop(\vec{\id_{\C_i \As}})
\end{equation*}
Which by~\eqref{eq:kleisli-lift} must be a~$\CD$-morphism of type~$\D(\op(\vec{\As_i})) \rightarrow \op(\vec{\C_i(\As_i)})$. Noting~$\klop(\vec{f_i \circ \counit}) = \op(\vec{f_i}) \circ \counit$, we establish naturality as follows:
\begin{align*}
    \kappa \circ \D\op(\vec{f_i}) &= \klop(\vec{\id}) \circ \D\op(\vec{f_i})\\
    &= \klop(\vec{\id}) \circ \D\op(\vec{f_i}) \circ \D(\counit) \circ \delta\\
    &= \klop(\vec{\id}) \circ \D\op(\vec{f_i \circ \counit}) \circ \delta\\
    &= \klop(\vec{\id}) \klcomp \klop(\vec{f_i \circ \counit})\\
    &= \klop(\vec{\id \klcomp (f_i \circ \counit)})\\
    &= \klop(\vec{\id \circ \C_i(f_i) \circ \C_i(\counit) \circ \delta})\\
    &= \klop(\vec{\C_i(f_i)})\\
    &= \klop(\vec{\C_i(f_i) \circ \counit \circ \C_i(\id) \circ \delta})\\
    &= \klop(\vec{\C_i(f_i) \circ \counit)} \klcomp \klop(\id)\\
    &= \klop(\vec{\C_i(f_i)}) \circ \counit \klcomp \klop(\id)\\
    &= \op(\vec{\C_i(f_i)}) \circ \counit \circ \D\klop(\id) \circ \delta\\
    &= \op(\vec{\C_i(f_i)}) \circ \klop(\id) \circ \counit \circ \delta\\
    &= \op(\vec{\C_i(f_i)}) \circ \kappa
\end{align*}
For axiom~\eqref{ax:kl-law-counit}:
\begin{align*}
    \op(\counit) \circ \kappa &= \op(\counit) \circ \klop(\vec{\id})\\
    &= \op(\counit) \circ \klop(\vec{\id}) \circ \counit \circ \delta\\
    &= \op(\counit) \circ \counit \circ \D\klop(\vec{\id}) \circ \delta\\
    &= \klop(\counit \circ \counit) \circ \D\klop(\vec{\id}) \circ \delta\\
    &= \klop(\counit \circ \counit) \klcomp \klop(\vec{\id})\\
    &= \klop(\vec{\counit \circ \counit \klcomp \id})\\
    &= \klop(\vec{\counit \circ \counit \circ \C_i(\id) \circ \delta})\\
    &= \klop(\vec{\counit \circ \counit \circ \delta})\\
    &= \klop(\vec{\counit})\\
    &= \counit
\end{align*}
For axiom~\eqref{ax:kl-law-comultiplication}:
\begin{align*}
    \kappa \circ \D(\kappa) \circ \delta &= \klop(\vec{\id}) \circ \D\klop(\vec{\id}) \circ \delta \\
    &= \klop(\vec{\id}) \klcomp \klop(\vec{\id})\\
    &= \klop(\vec{\id \circ \id})\\
    &= \klop(\vec{\id \circ \C_i(\id) \circ \delta})\\
    &= \klop(\vec{\delta})\\
    &= \klop(\vec{\delta \circ \counit \circ \C_i(\id) \circ \delta})\\
    &= \klop(\vec{\delta \circ \counit \klcomp \id})\\
    &= \klop(\vec{\delta \circ \counit}) \klcomp \klop(\vec{\id})\\
    &= \klop(\vec{\delta \circ \counit}) \circ \D\klop(\vec{\id}) \circ \delta\\
    &= \op(\vec{\delta}) \circ \counit \circ \D\klop(\vec{\id}) \circ \delta\\
    &= \op(\vec{\delta}) \circ \klop(\vec{\id}) \circ \counit \circ \delta\\
    &= \op(\vec{\delta}) \circ \klop(\vec{\id})\\
    &= \op(\vec{\delta}) \circ \kappa
\end{align*}
Finally, we must check the two mappings exhibit a bijection. For~$\klop$ satisfying~\eqref{eq:kleisli-lift}:
\begin{align*}
    \op(\vec{f_i}) \circ \klop(\vec{\id}) &= \op(\vec{f_i}) \circ \klop(\vec{\id}) \circ \counit \circ \delta\\
    &= \op(\vec{f_i}) \circ \counit \D\klop(\vec{\id}) \circ \delta\\
    &= \klop(\vec{f_i \circ \counit}) \circ \D\klop(\vec{\id}) \circ \delta\\
    &= \klop(\vec{f_i \circ \counit}) \klcomp \klop(\vec{\id})\\
    &= \klop(\vec{(f_i \circ \counit) \klcomp \id})\\
    &= \klop(\vec{f_i \circ \counit \circ \C_i(\id) \circ \delta})\\
    &= \klop(\vec{f_i \circ \counit \circ \delta})\\
    &= \klop(\vec{f_i})
\end{align*}
For~$\klop$ induced by a Kleisli-law:
\begin{align*}
    \klop(\vec{\id}) &= \op(\vec{id}) \circ \kappa
    = \kappa
    \qedhere
\end{align*}
\endgroup
\end{proof}

\KleisliLogicalEquivalences*
\begin{proof}
For preservation of~$\leftrightarrows$, let~$\klop$ be the lifted functor induced by~$\kappa$. If for~$1 \leq i \leq n$, $\As_i \leftrightarrows_{\C_i} \Bs_i$, then equivalently there are $\Klei{\C_i}$-morphisms~$f_i : \KLF{\C_i}(\As_i) \rightarrow \KLF{\C_i}(\Bs_i)$ and~$g_i : \KLF{\C_i}(\Bs_i) \rightarrow \KLF{\C_i}(\As_i)$. Applying~$\klop$ yields $\Klei{\D}$-morphisms:
\begin{align*}
\klop(\vec{\KLF{\C_i}(\As_i)}) &\xrightarrow{\klop(\vec{f_i})} \klop(\vec{\KLF{\C_i}(B_i)}) \quad\mbox{ and} \\
\klop(\vec{\KLF{\C_i}(B_i)}) &\xrightarrow{\klop(\vec{f_i})} \klop(\vec{\KLF{\C_i}(\As_i)})
\end{align*}
Then we note that $\klop(\vec{\KLF{\C_i}(\As_i)}) = \KLF{\D}(\op(\vec{\As_i}))$ and similarly for the $B_i$'s, therefore~$\op(\vec{\As_i}) \leftrightarrows_{\D} \op(\vec{\Bs_i})$.

The argument for preservation of~$\cong_{\KleiEmpty}$ is similar, but in this case we work with isomorphisms, which are preserved by functoriality.
\end{proof}

\subsection{Proofs for section~\ref{s:lifting-ops-to-coalgs}}
The following result appears in various guises in the literature~\cite{jacobs1994semantics, seal2013}, although with varying assumptions. We provide a full proof at the level of generality we require.
\EMLifting*
\begin{proof}
We note that the parallel pair:
\[
    \begin{tikzcd}[column sep=2.5em]
        \EMF{\D}(\op(\vec{A_i}))
            \rar[yshift=0.5em]{\EMF{\D}(\kappa)\circ \delta}
            \rar[swap,yshift=-0.5em]{\EMF{\D}(\op(\vec{\alpha_i}))}
        & \EMF{\D}(\op(\vec{\C_i(\As_i)}))
    \end{tikzcd}
\]
is reflexive, with mutual post-inverse~$\D \op(\counit)$.

We define the action on morphisms of~$\lop$ by the universal property of equalizers, as in the following diagram:
\[
    \begin{tikzcd}[column sep=2.5em]
    \lop(\vec{\alpha_i}) \dar[dashed,swap]{\klop(f_i)} \rar{\iota_{\vec{\alpha_i}}} & \EMF{\D}(\op(\vec{\As_i})) \dar[swap]{\EMF{\D}(\op(\vec{f_i}))} & \\
        \lop(\vec{\beta_i})
            \rar[swap]{\iota_{\vec{\beta_i}}}
        & \EMF{\D}(\op(\vec{\Bs_i}))
            \rar[yshift=0.5em]{\EMF{\D}(\kappa)\circ \delta}
            \rar[swap,yshift=-0.5em]{\EMF{\D}(\op(\vec{\beta_i}))}
        & \EMF{\D}(\op(\vec{\C_i(\Bs_i)}))
    \end{tikzcd}
\]
The parallel pair is equalized by the following calculation:
\begin{align*}
    \EMF{\D}(\op(\vec{\beta_i})) \circ \EMF{\D}(\op(\vec{f_i})) \circ \iota_{\vec{\alpha_i}} &= \D\op(\vec{\beta_i \circ f_i}) \circ \iota_{\vec{\alpha_i}} \\
    &= \D\op(\vec{\C_i(f_i) \circ \alpha_i}) \circ \iota_{\vec{\alpha_i}}\\
    &= \D\op(\vec{\C_i(f_i)}) \circ \D\op(\vec{\alpha_i}) \circ \iota_{\vec{\alpha_i}}\\
    &= \D\op(\vec{\C_i(f_i)}) \circ \D(\kappa) \circ \delta \circ \iota_{\vec{\alpha_i}}\\
    &= \D(\op(\vec{\C_i(f_i)}) \circ \kappa) \circ \delta \circ \iota_{\vec{\alpha_i}}\\
    &= \D(\kappa \circ \D\op(\vec{f_i})) \circ \delta \circ \iota_{\vec{\alpha_i}}\\
    &= \D(\kappa) \circ \D^2(\op(\vec{f_i})) \circ \delta \circ \iota_{\vec{\alpha_i}}\\
    &= \D(\kappa) \circ \delta \circ \D(\op(\vec{f_i})) \circ \iota_{\vec{\alpha_i}}\\
    &= \EMF{\D}(\kappa) \circ \delta \circ \EMF{\D}(\op(\vec{f_i})) \circ \circ \iota_{\vec{\alpha_i}}
\end{align*}
As~$\id_{\lop(\vec{\alpha_i})}$ makes the following commute:
\[
    \begin{tikzcd}[column sep=2.5em]
    \lop(\vec{\alpha_i}) \dar[dashed,swap]{} \rar{\iota_{\vec{\alpha_i}}} & \EMF{\D}(\op(\vec{\As_i})) \dar{\EMF{\D}(\op(\vec{\id}))} \\
        \lop(\vec{\alpha_i})
            \rar[swap]{\iota_{\vec{\alpha_i}}}
        & \EMF{\D}(\op(\vec{\As_i})) 
    \end{tikzcd}
\]
By the equalizer universal property~$\lop(\vec{\id}) = \id$. The following diagram commutes:
\[
    \begin{tikzcd}[column sep=2.5em]
    \lop(\vec{\alpha_i}) \dar[dashed,swap]{\lop(\vec{f_i})} \rar{\iota_{\vec{\alpha_i}}} & \EMF{\D}(\op(\vec{\As_i})) \dar{\EMF{\D}(\op(\vec{f_i}))} \\
    \lop(\vec{\beta_i}) \dar[dashed, swap]{\lop(\vec{g_i})}
        \rar[swap]{\iota_{\vec{\beta_i}}} & 
        \EMF{\D}(\op(\vec{\Bs_i})) \dar{\EMF{\D}(\op(\vec{g_i}))}\\
    \lop(\vec{\gamma_i})
        \rar[swap]{\iota_{\vec{\gamma_i}}} & \EMF{\D}(\op(\vec{\Cs_i}))
    \end{tikzcd}
\]
Therefore by the universal property of equalizers, $\lop(\vec{g_i \circ f_i}) = \lop(\vec{g_i}) \circ \lop(\vec{f_i})$. Therefore~$\lop$ is a well-defined functor. Note also that by construction~$\iota$ is natural of type~$\lop \rightarrow \EMF{\D} \circ H \circ \prod_i \EMU{\C_i}$.

The composite~$\D(\kappa) \circ \delta$ is a morphism of type~$\EMF{\D}(\op(\vec{\As_i})) \rightarrow \EMF{\D}(\op(\vec{\C_i(\As_i)}))$ as:
\[ \delta \circ \D(\kappa) \circ \delta = \D^2(\kappa) \circ \delta \circ \delta = \D^2(\kappa) \circ \D(\delta) \circ \delta = \D(\D(\kappa) \circ \delta) \circ \delta \]
We also note~$\D(\kappa) \circ \delta$ equalizes $\EMF{\D}(\kappa) \circ \delta$ and~$\EMF{\D}(\op(\vec{\delta}))$ as:
\begin{align*}
    \D(\kappa) \circ \delta \circ \D(\kappa) \circ \delta &= \D(\kappa) \circ \D^2(\kappa) \circ \delta \circ \delta\\ 
    &= \D(\kappa) \circ \D^2(\kappa) \circ \D(\delta) \circ \delta \\
    &= \D\op(\vec{\delta}) \circ \D(\kappa) \circ \delta
\end{align*}
Assuming~$f$ equalizes~$\EMF{\D}{\kappa} \circ \delta$ and~$\EMF{\D}(\op(\vec{\delta}))$, then:
\begin{align*}
    \D(\kappa) \circ \delta \circ \D\op(\vec{\counit}) \circ f
    &= \D(\kappa) \circ \D^2(\op(\vec{\counit})) \circ \delta \circ f \\
    &= \D(\kappa \circ \D\op(\vec{\counit})) \circ \delta \circ f \\
    &= \D(\op(\C_i \counit) \circ \kappa) \circ \delta \circ f \\
    &= \D(\op(\C_i \counit)) \circ \D(\kappa) \circ \delta \circ f \\
    &= \D(\op(\C_i \counit)) \circ \D(\op(\vec{\delta})) \circ f \\
    &= f
\end{align*}
$\D(\kappa) \circ \delta$ is split mono as:
\[ \D\op(\vec{\counit})) \circ \D(\kappa) \circ \delta = \D(\counit) \circ \delta = \id \]
Therefore the following is an equalizer diagram:
\[
\begin{tikzcd}[column sep=2.5em]
        \lop(\vec{\delta})
            \rar[swap]{\D(\kappa) \circ \delta}
        & \EMF{\D}(\op(\vec{\C_i(\As_i)}))
            \rar[yshift=0.5em]{\EMF{\D}(\kappa)\circ \delta}
            \rar[swap,yshift=-0.5em]{\EMF{\D}(\op(\vec{\delta}))}
        & \EMF{\D}(\op(\vec{\C^2_i(\As_i)}))
    \end{tikzcd}
\]
with~$\D\op(\vec{\counit}) \circ f$ the unique fill-in for a given~$f$. Therefore there is an induced isomorphism:
\[ \lop(\vec{\EMF{\C}(\As_i)}) \xrightarrow{\tilde{\iota_{\vec{\delta}}}} \EMF{\D}(\op(\vec{\As_i})) . \]
Then we calculate:
\begin{align*}
    \D(\kappa) \circ \delta \circ \EMF{\D}(\op(\vec{f_i})) \circ \tilde{\iota} &= \delta \circ  \EMF{\D}(\op(\vec{\C_i(f_i)})) \circ \D(\kappa) \circ \delta \circ \tilde{\iota}\\
    &= \delta \circ  \EMF{\D}(\op(\vec{\C_i(f_i)})) \circ \iota\\
    &= \iota \circ \lop(\vec{\EMF{\C_i}(f_i)})\\
    &= \D(\kappa) \circ \delta \circ \tilde{\iota} \circ \lop(\vec{\EMF{\C_i}(f_i)})
\end{align*}
Therefore~$\tilde{\iota}$ is natural of type~$\lop \circ \prod_i \C_i \Rightarrow \EMF{\D} \circ H$.
\end{proof}

\EMequalisers*
\begin{proof}
    We prove the more general statement, when $\C$ is on a wide subcategory $\CC$ of $\Rel$. Let $(\Ec,\Mc)$ be the (usual) proper factorisation system on $\Rel$.

    We first observe that $(\Ec\cap \CC, \Mc\cap \CC) = (\Ec\cap \CC, \Mc)$ is a factorisation system on $\CC$. Let $f\in \CC$ and $m \circ e$ be is its factorisation in $\Rel$. Because $\CC$ contains embeddings, $m\in \CC$ and by the last assumption we also have that $e\in \CC$ as $f=m\circ e$ is in $\CC$. The diagonality condition of factorisation systems is immediate. Consequently, $\EM{\CC}$ is equipped with the factorisation system $(\ol{\Ec\cap \CC}, \ol \Mc)$ by lemma~\ref{l:lifting-EMfs}.

    Next, we observe that $\CC$ is closed under equalisers and coproducts in $\Rel$. For the former, let $A \xrightarrow{e} B \overset{f,g}{\rightrightarrows} C$ be an equaliser diagram in $\Rel$ with $f,g\in \CC$. By lemma~\ref{l:fs-basics}.3, $e\in \CC$. Furthermore, let $D \xrightarrow{h} B \overset{f,g}{\rightrightarrows} C$ be a commutative diagram in $\CC$ and $\ol h\colon D \to A$ the unique morphism such that $e \circ \ol h = h$. Then, $\ol h \in \CC$ by our last assumption, and therefore, $e$ also an equaliser morphism of $f,g$ in $\CC$.

    Observe that $\EM{\C}$ is $\ol\Mc$-well-powered because $\CC$ is $\Mc$-well-powered because $\Rel$ is as well. The rest follows from proposition~\ref{p:abstract-linton} below and from the fact that $\CC$ has coproducts and the forgetful functor $\EM{\C} \to \CC$ reflects them.
\end{proof}

The following is essentially an abstract form of Linton's theorem~\cite{linton1969coequalizers}, phrased using factorization systems to fit our needs.
\begin{proposition}
    \label{p:abstract-linton}
    If category~$\CC$:
    \begin{enumerate}
        \item Has a proper factorization system $(\Ec,\Mc)$.
        \item Is $\Mc$-well-powered.
        \item Has coproducts.
    \end{enumerate}
    Then~$\CC$ has equalizers.
\end{proposition}
\begin{proof}
    Consider parallel pair:
    \[
    \begin{tikzcd}
    \As
    \rar[yshift=0.5em]{f}
    \rar[swap,yshift=-0.5em]{g}
    & \Bs
    \end{tikzcd}
    \]
    As~$\CC$ is well-powered and has a factorization system, there is a set~$I$ of isomorphism classes of embeddings~$m_i : \struct{M}_i \rightarrow \As$ equalizing~$f$ and~$g$. As
    \[ f \circ [m_i] = [f \circ m_i] = [g \circ m_i] = g \circ [m_i] \]
    the following commutes:
    \[
    \begin{tikzcd}
    \coprod_i \struct{M}_i \rar{[m_i]} \dar[swap]{[m_i]} & \As \dar{g} \\
    \As \rar[swap]{f} & \Bs
    \end{tikzcd}
    \]
    As~$\CC$ has a factorization system, the morphism~$[m_i]$ factors as:
    \[
    \begin{tikzcd}
    \coprod_i \struct{M}_i \rar[twoheadrightarrow]{e} & \struct{E} \rar[rightarrowtail]{m} & A 
    \end{tikzcd}
    \]
    Assume~$h : \Cs \rightarrow \As$ equalizes~$f$ and~$g$. There must exist~$i$ such that~$h$ factors as~$m_i \circ e'$. Let~$j$ be the corresponding coproduct injection. The following commutes:
    \[
    \begin{tikzcd}
    \coprod_i \struct{M}_i \rar[twoheadrightarrow]{e} & \struct{E} \rar[rightarrowtail]{m} & \As \rar[yshift=0.5em]{f} \rar[swap,yshift=-0.5em]{g}& \Bs \\
    \struct{M}_i \uar{j} & \Cs \lar{e'} \urar[swap]{h} & &
    \end{tikzcd}
    \]
    As the factorization system is proper, $m$ is a monomorphism, and so~$e \circ k \circ e'$ is the unique fill in map, and the following is an equalizer diagram:
    \[
    \begin{tikzcd}
    \struct{E} \rar[rightarrowtail]{m} &
    \As
    \rar[yshift=0.5em]{f}
    \rar[swap,yshift=-0.5em]{g}
    & \Bs
    \end{tikzcd}
    \qedhere
    \]
\end{proof}

\subsection{Proofs for section~\ref{s:smooth-ops}}

\SmoothKleisliPreservesEquivalence*
\begin{proof}
    For $i=1,\dots,n$, the relation $A_i \bisim_{\C_i} B_i$ is witnessed by a span of open pathwise-embeddings
    \[ \EMF{\C_i}(\As_i) \leftarrow R_i \rightarrow \EMF{\C_i}(B_i),\]
    in $\EM{\C_i}$. Then, by definition,
    \[
        \lop(\vec{\EMF{\C_i}(\As_i)}) \leftarrow \lop(\vec{R_i}) \rightarrow \lop(\vec{\EMF{\C_i}(B_i)})
    \]
    is a span of open pathwise-embeddings. Furthermore, by proposition~\ref{p:kl-law}, we get
    $\EMF{\C}(\op(\vec{A_i})) \cong \lop(\vec{\EMF{\C_i}(\As_i)})$ and, similarly, $\EMF{\C}(\op(\vec{B_i})) \cong \lop(\vec{\EMF{\C_i}(B_i)})$.
\end{proof}

\EkSmoothOpPreserve*
\begin{proof}
By proposition \ref{p:ef-logic-2}, $\As \equiv_{\FO_k} \Bs$ is characterized by the existence of a span of open pathwise embeddings  
    \[
        \EMF{\Ek}(\trI\As) \leftarrow R \rightarrow \EMF{\Ek}(\trI\Bs)
    \]
Therefore, it suffices to check that given an $\FO$-$\Ek$-smooth operation $\op\colon\prod_{i} R(\sg_i) \rightarrow R(\tau)$ and spans of open pathwise embeddings:
\[
        \EMF{\Ek}(\trI\As_i) \xleftarrow{f_i} R_i \xrightarrow{g_i} \EMF{\Ek}(\trI\Bs_i) \quad\text{in } \EM{\Ek} \text{ for $\Ek$ over $\R(\sg_i)$}
\]
for all $1 \leq i \leq n$, we can produce another span of open pathwise embeddings:
\begin{equation}
\label{eq:op-span} 
        \EMF{\Ek}(\trI(H(\vec{A_i})) \leftarrow R' \rightarrow \EMF{\Ek}(\trI(H(\vec{B_i})) \quad\text{in } \EM{\Ek} \text{ for $\Ek$ over $\R(\tau)$}
\end{equation}
By assumption, $H$ is a $\FO$-$\Ek$-smooth operation, so there exists a natural transformation $\theta\colon \op \circ \trI \rightarrow \trI \circ \op$ such that the components of $\EMF{\Ek}(\theta)$ are open pathwise embeddings. This allows us to compute (\ref{eq:op-span}) as the following composition in $\EM{\Ek}$ for $\Ek$ over $\R(\tau)$: 
\begin{equation}
\label{eq:op-span-comp} 
\begin{tikzcd}
    & \lop(\vec{R_i}) \ar[dl,"\lop(\vec{f_i})"'] \ar[dr,"\lop(\vec{g_i})"] & \\ 
    \lop(\vec{\EMF{\Ek}(\trI(\As_i))} \ar[d,"\cong"'] & & \lop(\vec{\EMF{\Ek}(\trI(\As_i))} \ar[d,"\cong"] \\
    \EMF{\Ek}(\op(\trI(\vec{\As_i})) \ar[d,"\EMF{\Ek}(\theta)"'] & & \EMF{\Ek}(\op(\trI(\vec{\Bs_i})) \ar[d,"\EMF{\Ek}(\theta)"] \\
    \EMF{\Ek}(\trI(\op(\vec{\As_i})) & &  \EMF{\Ek}(\trI(\op(\vec{\Bs_i}))
\end{tikzcd}
\end{equation}
By proposition \ref{p:smooth-Kleisli-law}, $\lop(\vec{f_i})$ and $\lop(\vec{g_i})$ are open pathwise embeddings. By proposition \ref{p:kl-law}, we get that $\lop(\vec{\EMF{\Ek}(\trI(\As_i))} \cong \EMF{\Ek}(\op(\trI(\vec{\As_i}))$ and $\lop(\vec{\EMF{\Ek}(\trI(\Bs_i))} \cong \EMF{\Ek}(\op(\trI(\vec{\Bs_i}))$. Moreover, $\EMF{\Ek}(\theta_{\vec{\As_i}})$ and $\EMF{\Ek}(\theta_{\vec{\Bs_i}})$ are open pathwise embedding. Since the composition of open pathwise embeddings is an open pathwise embedding \cite{AbramskyR21}, we conclude that the span in (\ref{eq:op-span-comp}) is a span of open pathwise embeddings as desired. 
\end{proof}

\subsection{Proofs for section~\ref{s:sufficient-axioms}}

\begin{lemma}
    \label{l:I1}
    If $e_1,\dots,e_n$ are embeddings in $\EM{\C_1}$, \dots, $\EM{\C_n}$, respectively, then so is $\lop(e_1,\dots,e_n)$.
\end{lemma}
\begin{proof}
    Recall that $\lop(e_1,\dots,e_n)$ is defined by the universal property of equalisers, while making the following diagram commute.
    \[
        \begin{tikzcd}
            \lop(\vec{\Ac_i})\dar[swap]{\lop(\vec{e_i})} \rar{\iota_{\vec{\Ac_i}}} & \EMF{\D}(\vec{\As_i}) \dar{\EMF{\D}(\vec{e_i})} \\
            \lop(\vec{\Ac'_i}) \rar{\iota_{\vec{\Ac'_i}}} & \EMF{\D}(\vec{\As'_i})
        \end{tikzcd}
    \]
    From $e_1,\dots,e_n$ being embeddings we know that $\EMF{\D}(\vec{U(e_i)})$ is an embedding as well. Furthermore, by lemma~\ref{l:fs-basics}(3), $\lop(\vec{e_i})$ is an equaliser and hence an embedding too. Consequently, $\lop(e_1,\dots,e_n)$ is an embedding by lemma~\ref{l:fs-basics}(2).
\end{proof}

\PEPreserved*
\begin{proof}
    The first part of the statement is proved in lemma~\ref{l:I1} above. For the second part, assume $e\colon \Pc \embed \lop(\vec{\Ac_i})$ is a path embedding where, for $i=1,\dots,n$, $\Ac_i$ is the domain of $f_i\colon \Ac_i \to \Bc_i$. By \ref{ax:s2}, $e$ decomposes as $e_0\colon \Pc \embed \lop(\vec{\Pc_i})$ followed by $\lop(\vec{e_i})\colon \lop(\vec{\Pc_i}) \embed \lop(\vec{\Ac_i})$. Observe that, by lemma~\ref{l:fs-basics}(2), $e_0$ is an embedding because $e$ is. Further, since $f_i$ is a pathwise-embedding, for $i=1,\dots,n$, the morphism $f_i \circ e_i$ is an embedding. Therefore, by lemma~\ref{l:I1}, $\lop(\vec{f_i \circ e_i})$ is also an embedding. We obtain that the composite $\lop(\vec{f_i}) \circ e = \lop(\vec{f_i \circ e_i}) \circ e_0$ is an embedding because embeddings are closed under composition.
\end{proof}

\begin{restatable}{lemma}{LemmaIFour}
    \label{l:I4}
    Assume $f_i\colon \Ac_i \to \Bc_i$ in $\EM{\C_i}$ are pathwise-embeddings, for $i=1,\dots,n$. Then, for any path embeddings $e$ and $g$ making the diagram on the left below commute
    \[
    \begin{tikzcd}[ampersand replacement=\&]
        \& \Pc \ar[>->,swap]{dl}{e} \ar[>->]{dr}{g} \\
        \lop(\vec{\Ac_i}) \ar{rr}{\lop(\vec{f_i})} \& \& \lop(\vec{\Bc_i})
    \end{tikzcd}
    \qquad\qquad
    \qquad\qquad
    \begin{tikzcd}[ampersand replacement=\&]
        \Pc_i \rar{f'_i}\dar[swap,>->]{e_i} \& \Qc_i\dar[>->]{g_i}\\
        \Ac_i \rar{f_i} \& \Bc_i
    \end{tikzcd}
    \]
    there exist morphisms $f'_i\colon \Pc_i \to \Qc_i$, for $i=1,\dots,n$, such that
    the diagram on the right above commutes.
    Here, the $e_i$ and $g_i$ are the minimal embeddings such that $e$ and $g$ decompose through $\lop(\vec{e_i})$ and $\lop(\vec{g_i})$, respectively.
\end{restatable}
\begin{proof}
    Let $e_0\colon \Pc \embed \lop(\vec{\Pc_i})$ be the embedding such that $\lop(\vec{e_i}) \circ e_0$ is the minimal decomposition of $e$. Since $\lop(\vec{f_i}) \circ e = \lop(\vec{f_i \circ e_i}) \circ e_0$ is another decomposition of $g$, there exists a morphism $l_i\colon \Qc_i \embed \Pc_i$ such that $g_i = f_i \circ e_i \circ l_i$, for $i=1,\dots,n$.

    Next, we observe that $e_0 = \lop(\vec{l_i}) \circ g_0$ where $g_0\colon \Pc \embed \lop(\vec{\Qc_i})$ is such that $\lop(\vec{g_i}) \circ g_0$ is the minimal decomposition of $g$. Observe that $\lop(\vec{f_i \circ e_i}) \circ e_0 = g = \lop(\vec{g_i}) \circ g_0 = \lop(\vec{f_i \circ e_i \circ l_i}) \circ g_0 = \lop(\vec{f_i \circ e_i}) \circ \lop(\vec{l_i}) \circ g_0$. Therefore, since $\lop(\vec{f_i \circ e_i})$ is an embedding (by lemma~\ref{l:I1}), we obtain $e_0 = \lop(\vec{l_i}) \circ g_0$.

    Consequently, we obtain another decomposition of $e$, given by $\lop(\vec{e_i \circ l_i}) \circ g_0$. By minimality $\vec{e_i}$, there exists $f'_i\colon \Pc_i \to \Qc_i$, for $i=1,\dots,n$, such that $e_i = e_i \circ l_i \circ f'_i$. Therefore, $f_i \circ e_i = f_i \circ e_i \circ l_i \circ f'_i = g_i \circ f'_i$, for $i=1,\dots,n$.
\end{proof}

\SmoothnessTheorem*
\begin{proof}
    Given open pathwise-embeddings $f_i\colon \Ac_i \to \Bc_i$ in $\EM{\C_i}$, for $i=1,\dots,n$, we need to check that $\lop(\vec{f_i})$ is open as we already know by proposition~\ref{p:pe-preserved} that it is a pathwise-embedding. Assume that the outer square of path embeddings in the diagram below commutes, with the left-most and right-most morphisms being the minimal decompositions of some path embeddings $\Pc \embed \lop(\vec{\Ac_i})$ and $\Qc \embed \lop(\vec{\Bc_i})$, respectively, by \ref{ax:s2}.
    \[
        \begin{tikzcd}[column sep=6.0em]
            \Pc
                \ar[>->]{rr}{h}
                \ar[>->,swap]{d}{e_0}
                \ar[>->]{dr}{e'_0}
            &
            & \Qc
                \ar[>->]{d}{g_0}
            \\
            \lop(\vec{\Pc_i})
                \ar[swap,>->]{d}{\lop(\vec{e_i})}
                \ar[dashed,swap]{r}{\lop(\vec{f'_i})}
            & \lop(\vec{\Pc'_i})
                \ar[sloped,>->,swap]{rd}{\lop(\vec{e'_i})}
                \rar[dashed]{\lop(\vec{h_i})}
            & \lop(\vec{\Qc_i})
                \ar[>->]{d}{\lop(\vec{g_i})}
            \\
            \lop(\vec{\Ac_i})
                \ar{rr}{\lop(\vec{f_i})}
            &
            & \lop(\vec{\Bc_i})
        \end{tikzcd}
    \]
    The path embedding $\lop(\vec{f_i\circ e_i}) \circ e_0$ has a minimal decomposition via $\lop(\vec{e'_i})\colon \lop(\vec{\Pc'_i}) \embed \lop(\vec{\Bc_i})$ as shown above. Then, by minimality of this decomposition and by lemma~\ref{l:I4}, for $i=1,\dots,n$, there exist $f'_i\colon \Pc_i \to \Pc'_i$ and $h_i\colon \Pc'_i \to \Qc_i$ such that $f_i \circ e_i = e'_i \circ f'_i$ and $e'_i = g_i \circ h_i$. Since $f_i$ is open and $f_i \circ e_i = g_i \circ h_i \circ f'_i$, there is a morphism $d_i\colon \Qc_i \to \Ac_i$ such that $d_i \circ h_i \circ f'_i = e_i$ and $g_i = d_i \circ f_i$. Finally, because the outer rectangle and the bottom rectangle commute and $\lop(\vec{g_i})$ is a mono, the top rectangle commute as well. Consequently, $\lop(\vec{d_i})\circ g_0\colon \Qc \to \lopveci\alpha$ is the required diagonal filler of the outer square.
\end{proof}

\subsection{Multilinear maps}
\label{s:multilin}

We introduce a class of maps between coalgebras of use when checking \ref{ax:s1} and \ref{ax:s2} axioms, and establish a connection to functors lifted to Eilenberg-Moore categories.

For a Kleisli-law~$\kappa$, $\D$-coalgebra $(\As,\alpha)$, and $\C_i$-coalgebras $(\Bs_i, \beta_i)$, we shall say that a $\CD$-morphism~$h : \As \rightarrow \op(\vec{\Bs_i})$ is a \df{multilinear map} $\Ac \to [\vec{\Bc_i}]$ if the following diagram commutes.

\[
    \begin{tikzcd}[->]
        \As \arrow[rr, "h"] \dar[swap]{\alpha} & & \op(\vec{\Bs_i}) \dar{\op(\vec{\beta_i})} \\
        \D(\As) \rar[swap]{\D(h)} & \D\op(\vec{\Bs_i}) \rar[swap]{\kappa} & \op(\vec{\C_i(\Bs_i)})
    \end{tikzcd}
\]

The following result appears in various guises in the literature for the dual setting of monads~\cite{jacobs1994semantics, seal2013}. We provide a full proof as similar results often appear restricted to the binary setting, and in relation to commutative monads. Our proof clarifies that the additional assumptions of a commutative (co)monad serve no role in the properties we require. We state the formulation suitable subsequent developments.

\begin{proposition}
    \label{p:multilin-maps-correspondence}
For a Kleisli-law~$\kappa$, let~$\iota$ be the universal equalizer morphism appearing in the construction of~$\lop$. The natural transformation:
    \[ \univ_{\vec{\Ac_i}} := \lop(\vec{\Ac_i}) \xrightarrow{\ee\iota} \D(\op(\vec{\As_i})) \xrightarrow{\ee\counit} \op(\vec{\As_i}) \]
is a \df{universal multilinear map} $\lop(\vec{\Ac_i}) \to [\vec{\Ac_i}]$.
    That is, any multilinear map~$f\colon \Bc \rightarrow [\vec{\Ac_i}]$ factors through $u$ via a unique ~$\D$-coalgebra morphism~$\lmulti f\colon \Bc \to \lop(\vec{\Ac_i})$.
\end{proposition}
\begin{proof}
We first establish~$\counit \circ \iota$ is multilinear. Let~$\omega$ be the structure map of~$\lop(\vec{\alpha_i})$
\begin{align*}
    \kappa \circ \D(\counit) \circ \D(\iota) \circ \omega &= \kappa \circ \D(\counit) \circ \delta \circ \iota \\
    &= \kappa \circ \counit \circ \delta \circ \iota \\
    &= \counit \circ \D(\kappa) \circ \delta \circ \iota\\
    &= \counit \circ \D\op(\vec{\alpha_i}) \circ \iota \\
    &= \op(\vec{\alpha_i}) \circ \counit \circ \iota
\end{align*}
Let~$(\Bs,\beta)$ be a coalgebra, and~$f : \beta \rightarrow [\vec{\alpha_i}]$ a multilinear map. Then~$\D(f) \circ \beta$ is a coalgebra morphism by:
\[ \D^2(f) \circ \D(\beta) \circ \beta = \D^2(f) \circ \delta \circ \beta = \delta \circ \D(f) \circ \beta . \]
Further, it equalizes~$\EMF{\D}(\op(\vec{\alpha_i}))$ and~$\EMF{\D}(\kappa) \circ \delta$, as:
\begin{align*}
    \D(\op(\vec{\alpha_i})) \circ \D(f) \circ \beta &= \D(\kappa) \circ \D^2(f) \circ \D(\beta) \circ \beta \\
    &= \D(\kappa) \circ \D^2(f) \circ \delta \circ \beta \\
    &= \D(\kappa) \circ \delta \circ \D(f) \circ \beta
\end{align*}
Let~$\tilde{f}$ be the unique coalgebra morphism given by the equalizer universal property by~$\D(f) \circ \beta$. 
Now if~$g : \beta \rightarrow \lop(\vec{\alpha_i})$, then~$\counit \circ \iota \circ g$ is a multilinear map as:
\begin{align*}
    \op(\vec{\alpha_i}) \circ \counit \circ \iota \circ g &= \kappa \circ \D(\counit \circ \iota) \circ \omega \circ g \\
    &= \kappa \circ \D(\counit \circ \iota) \circ \D(g) \circ \beta\\
    &= \kappa \circ \D(\counit \circ \iota \circ g) \circ \beta
\end{align*}
Then combining these two operations in one direction:
\[ \counit \circ \iota \circ \tilde{f} = \counit \circ \D(f) \circ \beta = f \circ \counit \circ \beta = f \]
In the other direction, we claim~$\widetilde{\counit \circ \iota \circ f} = f$. As:
\begin{align*}
    \D(\counit \circ \iota \circ f) \circ \beta &= \D(\counit) \circ \D(\iota) \circ \D(f) \circ \beta \\
    &= \D(\counit) \circ \D(\iota) \circ \omega \circ f\\
    &= \D(\counit) \circ \delta \circ \iota \circ f\\
    &= \iota \circ f
\end{align*}
The equalizer universal property then establishes the claimed equality. The two mappings then exhibit the required bijection.
\end{proof}

\begin{remark}
    Our multilinear map terminology is connected to an instance of the dual construction in monad theory. See for example~\cite{jacobs1994semantics, seal2013}, where the corresponding binary notion (called ``bilinearity'') is phrased in terms of Kleisli maps.
\end{remark}

\begin{lemma}
    \label{l:multi-decomp}
    Given a multilinear map $f\colon \Ac \to [\vec{\Bc_i}]$ which decomposes as $f_0\colon \As \to H(\vec{\As_i})$ followed by $H(\vec{e_i})$, for some coalgebra embeddings $e_i\colon \Ac_i \embed \Bc_i$, for $1 \leq i \leq n$, the morphism $f_0$ is a multilinear map $\Ac \to [\vec{\Ac_i}]$.
\end{lemma}
\begin{proof}
    We show that the left oblong in the diagram below commutes.
    \[
    \begin{tikzcd}[column sep=3.5em]
        \As
            \rar{f_0}
            \ar{dd}{\alpha}
        & \opveci\As
            \rar{\opveci e}
            \ar{d}{\opveci\alpha}
        & \opveci\Bs
            \ar{d}{\opveci \beta}
        \\
        & \opvec{\D(\As_i)}
            \rar{\opvec{\D(e_i)}}
        & \opvec{\D(\Bs_i)}
        \\
        \D(P)
            \rar{\D(f_0)}
        & \D(\opveci \As)
            \uar{\kappa}
            \rar{\D(\opveci e)}
        & \D(\opveci \Bs)
            \uar{\kappa}
    \end{tikzcd}
    \]
    By multilinearity of $f$ we know that the outer square commutes. Further, the two squares on the right commute by naturality of $\kappa$ and the assumption that $e_i$, for $1\leq i \leq n$, is a coalgebra morphism. A simple diagram chasing implies that $\opvec{\D(e_i)} \circ \opveci\alpha \circ f_0 = \opvec{\D(e_i)}\circ \kappa \circ \D(f_0)\circ \alpha$. Since $\D$ and $H$ preserve embeddings, $\opvec{\D(e_i)}$ is an embedding and hence a mono, proving multilinearity of $f_0$.
\end{proof}

Next, we translate axiom \ref{ax:s2} in terms of multilinear maps, which makes it easier to work with.

\begin{restatable}{lemma}{StwoMultilinear}
    \label{l:s2-multilin}
    Assuming \ref{ax:s1} holds, then \ref{ax:s2} is equivalent to
\begin{axioms}
    \item[\em\textbf{\namedlabel{ax:s2p}{(S2')}}]
    Any multilinear map $f\colon \pi \to [\vec{\Ac_i}]$, such that $\lmulti f$ is a path embedding,
    has a minimal decomposition through $\op(\vec{e_i})\colon \op(\vec{\Pc_i}) \to \op(\vec{\Ac_i})$, for some path embeddings $e_i\colon \Pc_i \embed \Ac_i$, with $1 \leq i \leq n$.
\end{axioms}
\end{restatable}
\begin{proof}
    Let $e\colon \pi \embed \lop(\vec{\Ac_i})$ be a path embedding and let $f\colon \pi \to [\vec{\Ac_i}]$ be the corresponding multilinear map (i.e.\ $f = \univ \circ e$).

    Assume $f$ has a minimal decomposition $\op(\vec{e_i})\circ f_0$ for some $f_0\colon \Ps \to \op(\vec{\pi_i})$ and path embeddings $e_i\colon \pi_i \embed \Ac_i$, with $1 \leq i \leq n$. It follows that $f_0$ is a multilinear map $\Pc \to [\vec{\Pc_i}]$ (see lemma~\ref{l:multi-decomp}). Therefore, by proposition~\ref{p:multilin-maps-correspondence}, there exists a coalgebra morphism $e_0\colon \pi \to \lop(\vec{\pi_i})$ such that $\univ \circ e_0 = f_0$. Then, by naturality of $\univ$, $f = \op(\vec{e_i}) \circ f_0 = \op(\vec{e_i}) \circ \univ \circ e_0 = \univ \circ \lop(\vec{e_i}) \circ e_0$ which, by universality of $\univ$, yields that $\lop(\vec{e_i}) \circ e_0 = e$ (and whereby $e_0$ is an embedding by lemma~\ref{l:fs-basics}). Further, if $e$ decomposes as $g_0\colon \Pc \to \lop(\vec{\Qc_i})$ followed by $\lop(\vec{g_i})$, with embeddings $g_i\colon \Qc_i \embed \Ac_i$, for $1\leq i \leq n$, then by naturality of $\univ$, $f$ decomposes as $H(\vec{g_i})\circ \univ \circ g_0$. Then, by minimality of $\op(\vec{e_i})\circ f_0$, there exists $g'_i\colon \Pc_i \to \Qc_i$ such that $e_i = g_i \circ g'_i$, for $1 \leq i \leq n$.

    Conversely, assuming $e$ has a minimal decomposition $\lop(\vec{e_i})\circ e_0$ for some $e_0\colon \pi \embed \lop(\vec{\pi_i})$ and path embeddings $e_i\colon \pi_i \embed \Ac_i$, with $1 \leq i \leq n$. By naturality of $\univ$, we obtain a decomposition of $f$ as
    \[ f = \univ \circ \lop(\vec{e_i}) \circ e_0 = \op(\vec{e_i}) \circ f_0 \]
    where $f_0 = \univ \circ e_0\colon \Ps \to \op(\vec{\Ps_i})$. If there is another decomposition of $f$ as $g_0\circ \op(\vec{\Qs_i})$ followed by $H(\vec{g_i})$, for embeddings $g_i\colon \Qc_i \to \Ac_i$ (where $1 \leq i \leq n$), then by lemma~\ref{l:multi-decomp}, $g_0$ is multilinear and by proposition~\ref{p:multilin-maps-correspondence} there exists $h\colon \pi \to \lop(\vec{\Qc_i})$ such that $g_0 = \univ \circ h$. As before we see that $e = \lop(\vec{g_i}) \circ h$, yielding $g'_i\colon \Pc_i \to \Qc_i$ such that $e_i = g_i \circ g'_i$, for $1 \leq i \leq n$
\end{proof}

\subsection{Proofs for section~\ref{s:FO-example-edge-creation}}

\etaRiSmooth*
\begin{proof}
    Axioms \ref{ax:kl-law-counit} and \ref{ax:kl-law-comultiplication} follow automatically because the underlying function of $\kappa$ is the identity function.

    Next, we check that $\kappa$ satisfies the smoothness axioms \ref{ax:s1} and \ref{ax:s2p} from lemma~\ref{l:s2-multilin}.
For \ref{ax:s1}, we see that $\etaRi$ preserves embeddings $e:\As \rightarrow \Bs$, as it follows that $R^{\etaRi(\As)}(a_1,\dots,a_r)$ if and only if $R^{\etaRi(\Bs)}(\etaRi(e)(a_1),\dots,\etaRi(e)(a_r))$ from $R^{\As}(a_1,\dots,a_r) \Leftrightarrow R^{\Bs}(e(a_1),\dots,e(a_r))$ and also $P^{\As}_{i_z}(a_z) \Leftrightarrow P^{\Bs}_{i_z}(e(a_z))$ for all $z \leq r$.

For \ref{ax:s2p}, suppose $\pi\colon \Ps \to \Ek \Ps$ is a path $r_1 \sqsubset_\pi r_2 \sqsubset_\pi \dots \sqsubset_\pi r_n$ and $f\colon \Pc \to [\Ac]$ is a multilinear map. Because $\kappa$ is just the identity function, multilinearity of $f$ translates as
\begin{align}
    \alpha(f(r_i)) = \kappa([f(r_1),\dots,f(r_n)]) = [f(r_1),\dots,f(r_n)], 
\end{align}
for every  $i=1,\dots,n$. In other words, $f(r_1), \dots, f(r_n)$ is a path in $\sqsubseteq_\alpha$ (recall remark~\ref{r:coalg-order}). Therefore, the inclusion $j\colon \Ps' \to \As$ of the induced substructure on $\{ f(r_1), \dots, f(r_n) \}$ into $\As$, is a path embedding $\Pc' \to \Ac$ where $\pi'$ is the restriction of $\alpha$ to $\Ps$.
Further, the mapping $f$ factors as $\etaRi(j)\circ f_0$ where $f_0$ is a homomorphism because $f$ is and $j$ is an embedding.
Minimality is immediate from the fact that the image of $\Ps$ under $f$ is onto $\Ps'$.
\end{proof}

\subsection{Proofs for section~\ref{s:FO-example-coproducts}}

We first justify that $\Eknc$ is a comonad. Viewing $\Eknc(A)$ as an induced substructure of $\Ek(A)$, we see that the restriction of $\counit\colon \Ek(A) \to A$ to $\Eknc(A)$ is a morphism in $\Rnc(\sg)$. Further, defining coextensions $(-)^*$ as for $\Ek$ we see that, for any morphism $f\colon \Eknc(A) \to B$ in $\Rnc(\sg)$, the coextension $f^*$ is also a morphism in $\Rnc(\sg)$. The axioms \eqref{eq:comonad-axioms} hold for $(\Eknc, \counit, (-)^*)$ because they already hold for $\Ek$.

Observe that $\EM{\Eknc}$ has equalisers of reflexive pairs. This follows from lemma~\ref{l:EM-equalisers}, it is easy to see that its assumptions are satisfied for $\C = \Eknc$ and $\CC = \Rnc(\sg_3)$. Therefore, if $\kappa$ is a Kleisli law, the functor $\plus$ lifts to a functor $\blift{\plus}\colon \EM{\Eknc}\times \EM{\Eknc} \to \EM{\Eknc}$. The following shows that this is the case and that the lifted functor preserves open pathwise embeddings.

\KappaCoproducts*
\begin{proof}
    We check that $\kappa$ is a well-defined natural transformation. First, to see that each component of $\kappa$ is a well-defined $\sg$-structure homomorphism, assume that $(x_1,\dots,x_n)\in R^{\Eknc(A\plus B)}$ holds for some $R$ in $\sg$ and $x_1,\dots,x_n$ in $\Eknc(A\plus B)$. By definition of $\Eknc(A\plus B)$, $\counit(x_i) = \iota(y_i)$, for some $(y_1,\dots,y_n)\in R^\As$ or $(y_1,\dots,y_n)\in R^\Bs$. W.l.o.g.\ assume that the former is the case.
    By definition of $\kappa$, each $\kappa(x_i)$ is either $\ex_1(x_i)\in \Eknc(A)$ or a constant $c_i$. Therefore, $\kappa(x_i) = \inc(z_i)$ for some $z_i \in \Eknc(A)$ and, moreover, $\counit(z_i) = y_i$, for $1\leq i \leq n$. Finally, we see that non-constant $x_i$ are comparable in the prefix order in $\Eknc(A\plus B)$ and this property is preserved by $\ex_1(-)$. Therefore, $(\kappa(x_1),\dots,\kappa(x_n))\in R^{\Eknc(A) \plus \Eknc(B)}$.

    Naturality of $\kappa$ is immediate from the fact that morphisms in $\Rnc(\sg_1)$ and $\Rnc(\sg_1)$ send constant elements to constant elements and non-constant elements to non-constant elements.

    Next, we check the smoothness axioms. Observe that \ref{ax:kl-law-counit} and \ref{ax:s1} hold by definitions. Further, \ref{ax:kl-law-comultiplication} follows by an easy case analysis.
    For \ref{ax:s2p} from lemma~\ref{l:s2-multilin}, let $f\colon \Pc \to [\Ac,\Bc]$ be a multilinear map from a path $\pi$ in $\EM{\Eknc}$ and let $r_1 \sqsubset_\pi r_2 \sqsubset_\pi \dots \sqsubset_\pi r_n$ be the chain of non-constant elements in $\Pc$. Multilinearity of $f$ implies
    \begin{equation}
        \kappa([f(r_1),\dots,f(r_i)]) = (\alpha\plus\beta)(f(r_i)),\quad \text{for $i=1,\dots,n$.}
    \end{equation}
    In fact, this expresses precisely that if $f(r_i)\in \inc(\As)$ then $\ex_1([f(r_1),\dots,f(r_i)])$ is a path in $\sqsubseteq_\alpha$ and, dually, $f(r_i)\in \inc(\Bs)$ implies that $\ex_2([f(r_1),\dots,f(r_i)])$ is a path in $\sqsubseteq_\beta$. From this it immediately follows that $f$ factors through the embedding
    \[ e_1\plus e_2\colon \Ps \plus \Qs \embed \As \plus \Bs  \]
    where $e_1$ and $e_2$ are the path embeddings corresponding to the paths $\ex_1([f(r_1),\dots,f(r_n)])$ and $\ex_2([f(r_1),\dots,f(r_n)])$, respectively. Minimality of $\Ps$ and $\Qs$ is immediate.

    We remark that $\kappa$ is smooth by the above discussion, theorem~\ref{t:smoothness} and lemma~\ref{l:s2-multilin}.
\end{proof}

\begin{lemma}
    \label{l:bisim-Eknc}
    For $\sg$-structures $A,B$,\ $\trI(A) \bisim_{\Eknc} \trI(B)$ iff $A \equiv_{\FO_k} B$
\end{lemma}
\begin{proof}
   The proof of proposition \ref{p:ef-logic-2} can be adapted to this statement. Moreover, we prove a general version for multi-sorted first order logic in proposition \ref{prop:many-sorted-ek-bisim}. 
\end{proof}

\FOlogicEquivCoproducts*
\begin{proof}
    We need to show that $\plus$, viewed as an operation $\R_0(\sg_1)\times \R_0(\sg_2) \to \R_0(\sg_3)$ is $\FO$-$\Eknc$-smooth. We defined $\Eknc$ as a comonad on $\Rnc(\sg)$, for any $\sg$, i.e.\ including $\sg = \sgI_1,\sgI_2$, or $\sgI_3$. Also, the usual $\trI\colon \R_0(\sg) \to \RelI$ factors as $\R_0(\sg) \to \Rnc(\sgI)$ followed by the inclusion $\Rnc(\sgI) \sue \RelI$. We need two natural transformations $\kappa$ and $\theta$ shown below.
    \begin{equation}
        \begin{tikzcd}[column sep=3.6em]
            \R_0(\sg_1) \times \R_0(\sg_1)
                \rar{\tr_1 \times \tr_2}
                \dar[swap]{\plus}
            & \Rnc(\sgI_1) \times \Rnc(\sgI_2)
                \dar[swap]{\plus\I}
                \ar[Rightarrow,swap,dl,"\theta"]
                \rar{\Eknc \times \Eknc}
            & \Rnc(\sgI_1) \times \Rnc(\sgI_2)
                \dar{\plus\I}
                \ar[Rightarrow,swap]{dl}{\kappa}
            \\
            \R_0(\sg_3)
                \rar[swap]{\tr_3}
            & \Rnc(\sgI_3)
                \rar[swap]{\Eknc}
            & \Rnc(\sgI_3)
        \end{tikzcd}
    \end{equation}
    Where $\plus\I$ is the $I$ version of $\plus$. It is easy to observe that $\trI(\As) \mathbin{\plus\I} \trI(\Bs) \cong \trI(\As \plus \Bs)$ and so $\theta$ is just this isomorphism, componentwise. Further, the $I$-version of $\kappa$ is smooth by lemma~\ref{l:kappa-coproducts}. 
    Therefore, by proposition~\ref{p:smooth-Kleisli-law}, if $\trI(A_i) \bisim_{\Eknc} \trI(B_i)$, for $i=1,2$, then also
    \[ \trI(A_1\plus A_2) \cong \trI(A_1)\plus\I\trI(A_2) \bisim_{\Eknc} \trI(B_1)\plus\I\trI(B_2) \cong \trI(B_1\plus B_2).\]
    This proves the statement as $\trI(A) \bisim_{\Eknc} \trI(B)$ is equivalent to $A \equiv_{\FO_k} B$, by lemma~\ref{l:bisim-Eknc}.
\end{proof}

\subsection{Proofs for section~\ref{s:Ek-multi-sort}}

\ManySortedEkIsComonad*
\begin{proof}
For $\As \in \RelS$, the morphism $\varepsilon_{\As}\colon\Ek\As \rightarrow \As$ is defined component-wise: 
\[ (\varepsilon_{\As})_s(c^{\Ek\As}) = c^{\As} \text{ for } c \in \consts{\sg}{s}\quad (\varepsilon_{\As})_s([a_1{:}s_1,\dots,a_n{:}s_n]) = a_n \text{ with } s = s_n\] 
For a $\sg$-morphism $f\colon\Ek\As \rightarrow \Bs$, the morphism $f^{*}\colon\Ek\As \rightarrow \Ek\Bs$ is defined component-wise:
\[ (f^{*})_s(c^{\Ek\As}) = c^{\Ek\Bs} \text{ for } c \in \consts{\sg}{s}\quad (f^{*})_s([a_1{:}s_1,\dots,a_n{:}s_n]) = [b_1{:}s_1,\dots,b_n{:}s_n]\] 
where $b_i = f_{s_i}([a_1{:}s_1,\dots,a_i{:}s_i])]$ for $1 \leq i \leq n$ and $s = s_n$.
Verifying that $\varepsilon_{\As}$ and $f^{*}$ are $\sg$-morphisms is straightforward application of the definitions. Moreover, verifying the axioms (\ref{eq:comonad-axioms}) is also routine. 

In order to show $\Eknc$ is comonad, we note that for all $\As \in \Rsnc(\sg)$, $\varepsilon_{\As}$ is a morphism in $\Rsnc(\sg)$ and that $f^{*}$ is a morphism in $\Rsnc(\sg)$ whenever $f$ is a morphism in $\Rsnc(\sg)$. 
\end{proof}

\ManySortedEkbisim*

We prove that bisimulation for the $\trJ$-lifting $S$-sorted $\Ek$ captures equivalence in $\FO_k(S,J)$, by first introducing the corresponding {\ef} game $\EF^{J}_k(\As,\Bs)$ and showing that this captures the equivalence in $\FO_k(S,J)$ by an adaptation of the standard single-sorted proof found in chapter three of \cite{libkin2004elements}. We then prove that the existence of a Duplicator winning strategy in $\EF^{J}_k(\As,\Bs)$ is equivalent to the existence of open pathwise embeddings $\EMF{\Ek} \trJ\As \leftarrow (W,\upsilon) \rightarrow \EMF{\Ek} \trJ\Bs$ by adapting the single-sorted proof found in \cite{abramsky2021relating}.
Later we explain how to adapt the proof to obtain also the second item of the statement.

Generalising the single-sorted case, every move of the $\EF^{J}_k(\As,\Bs)$ game results in a $S$-sorted relation $\gamma$ of type $A$ to $B$. A $S$-sorted relation $\gamma$ from $S$-sorted set $A$ to $S$-sorted set $B$ is an $S$-indexed family of sets $\gamma_{s} \subseteq A_s \times B_s$. Since the logic $\FO_k(S,J)$ has equality only for terms of sort $s \in J$, the winning condition of $\EF^{J}_k(\As,\Bs)$ requires a notion of $S$-sorted relation $\gamma$ which is only required to be single-valued, i.e.\ a partial function, on sorts $s \in J$. The following notion is used to define the winning moves in $\EF^{J}_k(\As,\Bs)$:
\begin{definition}
    Given $S$-sorted $\sg$ structures $\As$,$\Bs$ and $J \subseteq S$, a $S$-sorted relation $\gamma$ of type $A \times B$ is a \textit{partial isomorphism on $J$ from $\As$ to $\Bs$} if the following conditions hold:  
    \begin{itemize}
        \item If $c \in \consts{\sg}{s}$, then $(c^{\As},c^{\Bs}) \in \gamma_{s}$ 
        \item For all $s \in J$, $\gamma_{s}$ is a partial function from $A_s$ to $B_s$ 
        \item For every $R \in \sg$ of arity $\langle s_1,\dots,s_m \rangle$, 
        \[ R^{\As}(a_1,\dots,a_m) \Leftrightarrow R^{\Bs}(b_1,\dots,b_m) \]
        whenever $(a_1,b_1) \in \gamma_{s_1},\dots,(a_m,b_m) \in \gamma_{s_m}$ 
    \end{itemize}
\end{definition}

We are now able to introduce the $k$-round many-sorted {\ef} game:
 \begin{definition}
    Given $S$-sorted $\sg$ structures $\As$ and $\Bs$, the \textit{$k$-round $S$-sorted Ehrenfeucht-{\Fraisse} game $\EF^{J}_{k}(\As,\Bs)$ with equality on $J$} is played between two players, Spoiler and Duplicator. In each round $i \leq k$, 
    \begin{itemize}
        \item Spoiler chooses a sort $s_i \in S$ and element $a_i \in A_{s_i}$ or $b_i \in B_{s_i}$.
        \item Duplicator responds with an element $b_i \in B_{s_i}$  or $a_i \in A_{s_i}$ in the other structure's $s_i$ sort.
    \end{itemize}
    Duplicator wins round $i$, if the following $S$-sorted relation $\gamma$, defined componentwise,  
     \[ \gamma_{i,s} = \{(c^{\As},c^{\Bs}) \mid c \in \consts{\sg}{s}\} \cup \{(a_j,b_j) \mid \text{ for } s_j = s \text{ and } j \leq i\}\]
    is a partial isomorphism on $J$ from $\As$ to $\Bs$. \\
    
    \noindent Duplicator is said to have a \textit{winning strategy in $\EF^{J}_{k}(\As,\Bs)$} if for all possible Spoiler moves in the $k$-round game, Duplicator can win round $i$ for all $i \leq k$. 
\end{definition}   
As expected, the $\EF^{J}_k(\As,\Bs)$ game captures equivalence in $\FO_k(S,J)$. To prove this, we consider structures $(\As,a)$ in the expanded signature $\sg(c{:}s)$. The signature $\sg(c{:}s)$ has an additional fresh constant symbol $c$ added to $\consts{\sg}{s}$, the same constants as $\consts{\sg}{s'}$ for $s' \not= s$, and the same relation symbols as $\sg$. The notation $(\As,a)$ is the $\sg(c{:}s)$-structure where $\As$ is a $\sg$-structure and $c^{\As} = a \in A_s$. 
\begin{proposition}
\label{prop:bisim-game-logic}
The following are equivalent
\begin{enumerate}
    \item Duplicator has a winning strategy in $\EF^{J}_k(\As,\Bs)$ 
    \item $\As \equiv_{\FO_k(S,J)} \Bs$
\end{enumerate}
\end{proposition}
\begin{proof}
$(1) \Rightarrow (2)$ We prove the statement by induction on $k$. Suppose $k = 0$ and that Duplicator has a winning strategy in $\EF^{J}_{0}(\As,\Bs)$. By the winning condition, the relation $\gamma$, defined componentwise as ${\gamma_{s} = \{(c^{\As},c^{\Bs}) \mid c \in \consts{\sg}{s}\}}$ is a partial isomorphism. In particular, $\gamma$ is a partial isomorphism on $J$. Hence, from the definition of partial isomorphism, $\As$ and $\Bs$ satisfy the same atomic sentences. 

For the inductive step, suppose Duplicator has a winning strategy in $\EF^{J}_{k+1}(\As,\Bs)$ and consider the sentence $\phi$ of rank $k+1$. It follows from the definition of quantifier rank, that every sentence $\phi$ of rank $k+1$ is a boolean combination of sentences of the form $\exists y{:}s \psi(y{:}s)$ for sort $s \in S$ where $\psi$ is a formula of rank $\leq k$. Hence, by recursion on the construction of formulas $\phi(\vec{x})$, it suffices to only prove that $\As \vDash \phi \Leftrightarrow \Bs \vDash \phi$ for $\phi = \exists y{:}s \psi(y{:}s)$. 

Suppose $\As \vDash \exists y{:}s \psi(y{:}s)$, then there exists some $a \in A_s$, such that ${(\As,a) \vDash \psi(y{:}s)}$. Consider the $\EF^{J}_{k+1}(\As,\Bs)$ game where Spoiler plays $s \in S$ and $a \in A_s$ in the first round. Let Duplicator's response be $b \in B_s$ according to the winning strategy. Duplicator then has a winning strategy in the $\EF^{J}_{k}((\As,a),(\Bs,b))$ game by playing according the subsequent $k$ moves of the $\EF_{k+1}(\As,\Bs)$ game.  Hence, by the inductive hypothesis, we have that for all $\sg(c{:}s)$ sentences $\phi'$ of rank $\leq k$, $(\As,a) \vDash \phi' \Leftrightarrow (\Bs,b) \vDash \phi'$. In particular, $(\Bs,b) \vDash \psi(y{:}s)$. Hence, $\Bs \vDash \exists y{:}s \psi(y{:}s)$. To show the other direction, i.e. $\Bs \vDash \exists y{:}s \psi(y{:}s) \Rightarrow \As \vDash \exists y{:}s \psi(y{:}s)$, is essentially the same. 

$(2) \Rightarrow (1)$  We prove the statement by induction on $k$. For $k = 0$, $\As$ and $\Bs$ satisfying the same atomic sentences proves that $\gamma$, formed from just the pairs of constants, is a partial isomorphism on $J$. 

For the inductive step, assume $\As$ and $\Bs$ satisfy the same sentences of rank $k+1$. We need to show that Duplicator has a winning strategy in the $\EF_{k+1}(\As,\Bs)$ game. Suppose, in the first round of $k+1$, Spoiler chooses $s \in S$ and $a \in A_s$. There is finite collection of in-equivalent formulas $\varphi_{1},\dots,\varphi_{M}$ of rank $k$ with variable $x{:}s$ such that $(\As,a) \vDash \varphi_{i}(x{:}s)$. Hence, $\As \vDash \exists x{:}s(\bigwedge_{i \leq M} \varphi_{i}(x{:}s))$. The sentence $\exists x{:}s(\bigwedge_{i \leq M} \varphi_{i}(x{:}s))$ has rank $k+1$, so by the supposition, $\Bs \vDash \exists x{:}s(\bigwedge_{i \leq M} \varphi_{i}(x{:}s))$. Therefore, there exists a $b \in B_s$ witnessing the existential quantifier $\exists x{:}s$; and so Duplicator responds with this $b \in B_s$. By construction, ${(\As,a) \vDash \varphi(x{:}s) \Leftrightarrow (\Bs,b) \vDash \varphi(x{:}s)}$ for all $\sg(c)$-sentences of rank $\leq k$. By the inductive hypothesis, Duplicator has a winning strategy in the $\EF_{k}((\As,a),(\Bs,b))$ game. Therefore, Duplicator proceeding in the subsequent $k$ moves of $\EF_{k+1}(\As,\Bs)$ according to the moves of $\EF_{k}((\As,a),(\Bs,b))$ yields a winning strategy for $\EF_{k+1}(\As,\Bs)$. 
\end{proof}

\begin{proposition}
\label{prop:bisim-game-comonad}
The following are equivalent:
\begin{enumerate}
  \item Duplicator has a winning strategy in $\EF^{J}_k(\As,\Bs)$
  \item $\trJ\As \leftrightarrow_{\Ek} \trJ\Bs$  
\end{enumerate}
\end{proposition}
\begin{proof}
$(1) \Rightarrow (2)$ Suppose Duplicator has a winning strategy in $\EF^{J}_k(\As,\Bs)$, this determines a $S$-sorted set $W$ with components $W_{s} \subseteq (\Ek\As)_{s} \times (\Ek \Bs)_{s}$ that satisfy the following conditions:
\begin{itemize}
    \item For all $c \in \const{\sg,s}$, $(c^{\Ek \As}, c^{\Ek \Bs}) \in W_s$
    \item (forth) For every pair $(w,v) \in W_s$ such that $w$ is a word of length $<k$ or a pair of empty words $(w,v)$ and $a \in A_{s'}$ (Spoiler's move), there exists a $b \in B_{s'}$ (Duplicator's response) such that $(w[a{:}s'],v[b{:}s']) \in W_{s'}$.
    \item (back) For every pair $(w,v) \in W_s$ such that $v$ is a word of length $< k$ or a pair of empty words $(w,v)$ and $b \in B_{s'}$, there exists a $a \in A_{s'}$ such that 
    $(w[a{:}s'],v[b{:}s']) \in W_{s'}$
    \item (winning) The relation $\gamma_{w,v}$ where $(\gamma_{w,v})_s$ consists of pairs $(a,b)$ such that $a{:}s$ and $b{:}s$ appear at the same index in $w$ and $v$ (respectively) is partial isomorphism on $J$ for every pair $(w,v)$ in $S$-sorted set $W$. 
\end{itemize}
Since $W$ is a $S$-sorted subset of $\Ek\trJ\As \times \Ek \trJ\Bs$, we can consider the induced substructure $W$ of this product structure. This structure $W$ has a natural coalgebra $\omega:W \rightarrow \Ek W$ where constants are mapped to constants and $\omega_{s}(w_n,v_n)$ is mapped to the sequence $[(w_1,v_1){:}s_1,\dots,(w_n,v_n){:}s_n]$ where $w_1 \sqsubset_{\delta_\As} \dots \sqsubset_{\delta_\As} w_n$ and $v_1 \sqsubset_{\delta_\Bs} \dots \sqsubset_{\delta_\Bs} v_n$ are the unique chains below in $w_n$ and $v_n$ in $\EMF{\Ek}\trJ\As$ and $\EMF{\Ek}\trJ\Bs$, respectively.  
The projection maps $p_1:(W,\omega) \rightarrow \EMF{\Ek}\trJ\As$ and $p_2:(W,\omega) \rightarrow \EMF{\Ek}\trJ\Bs$ are coalgebra morphisms, but we must show that these projection maps are also open pathwise embeddings. 

To verify $p_1:(W,\omega) \rightarrow \EMF{\Ek}\trJ\As$ is a pathwise embedding, we note that the image of path embedding $e:(P,\pi) \rightarrow (W,\omega)$ contains elements $(w_1,v_1){:}s_1,\dots,(w_n,v_n){:}s_n$ where $w_1 \sqsubset_{\delta_\As} \dots \sqsubset_{\delta_\As} w_n$ and $v_1 \sqsubset_{\delta_\Bs} \dots \sqsubset_{\delta_\Bs} v_n$ are chains in $\Ek\trJ\As$ and $\Ek\trJ\Bs$, respectively. As $w_1 \sqsubset \dots \sqsubset w_n$ is a path in $\Ek\trJ\As$, we have that $p_1 \circ e:(P,\pi) \rightarrow \Ek\trJ\As$ is a path embedding.

To verify $p_1:(W,\omega) \rightarrow \EMF{\Ek}\trJ\As$ is open, consider the following commuting diagram in $\EM{\Ek}$:
\begin{center}
\begin{tikzcd}
    (P,\pi) \ar[r,rightarrowtail] \ar[d,rightarrowtail,"e_0"'] & (Q,\chi) \ar[d,rightarrowtail,"e_1"] \\ 
    (W,\omega) \ar[r,"\pi_1"] & \EMF{\Ek}\trJ\As
\end{tikzcd}
\end{center}
where $e_0:(P,\pi) \rightarrowtail (W,\omega)$ and $e_1:(Q,\chi) \rightarrowtail \EMF{\Ek}\trJ\As$ are path embeddings. Suppose $w_1 \sqsubset_{\delta_{\As}} \dots \sqsubset_{\delta_{\As}} w_n$ is the image of path $Q$ presented as $r_1 \sqsubset_{\chi} \dots \sqsubset_{\chi} r_n$ under embedding $e_1$, then by the above diagram commuting $(w_1,v_1) \sqsubset_{\omega} \dots \sqsubset_{\omega} (w_i,v_i)$ for some $i \leq n$ and $v_1 \sqsubset_{\delta_{\Bs}} \dots \sqsubset_{\delta_{\Bs}} v_i$ in $\Ek\trJ\Bs$. We need to construct an embedding $d:(Q,\chi) \rightarrowtail (W,\omega)$ that `fills the diagonal'.  

Suppose $w_j = w_i[a_{i+1}{:}s_{i+1},\dots,a_j{:}s_j]$ such that $w_j$ are the elements in the image of $e_1$, but not in the image of $e_1 \circ e$. By repeated applications of the forth property, we obtain pairs $(w_j,v_j) \in W_s$ with $v_j = v_i[b_{i+1}{:}s_{i+1},\dots,b_j{:}s_j]$. Let $d(r_z) = (w_z,v_z){:}s_z$ for all $z \leq n$. By the winning condition on $W_s$, $d$ is indeed an embedding. Hence, the following diagram commutes: 
\begin{center}
\begin{tikzcd}
    (P,\pi) \ar[r,rightarrowtail] \ar[d,rightarrowtail,"e_0"'] & (Q,\chi) \ar[d,rightarrowtail,"e_1"] \ar[dl,rightarrowtail,"d"] \\ 
    (W,\omega) \ar[r,"p_1"'] & \EMF{\Ek}\trJ\As
\end{tikzcd}
\end{center}

A similar proof shows that $p_2:(W,\omega) \rightarrow \EMF{\Ek}\Bs$ is a open pathwise embedding.

$(2) \Rightarrow (1)$ Suppose we have span of open pathwise embeddings $\EMF{\Ek}\trJ\As \xleftarrow{h_1} (W,\omega) \xrightarrow{h_2} \EMF{\Ek}\trJ\Bs$. We prove the following statement by induction on $n \leq k$: 

    $(*)$ For every $s \in S$ and $b \in B_s$, there exists a path $P$ consisting of constants $\const{P}$ and a chain $r_1 \sqsubset_{\pi} \dots \sqsubset_{\pi} r_n$ with embedding $d:(P,\pi) \rightarrowtail (W,\omega)$ and $a \in A_s$ such that: 
    \begin{itemize}
        \item $\varepsilon_{\As}(h_{1,s}(d_s(r_n))) = a$ and $\varepsilon_{\Bs}(h_{2,s}(d_s(r_n))) = b$ 
        \item The $S$-sorted relation $\gamma$ defined componetwise: 
        \begin{align*}
        \gamma_{n,t} &= \{(c^{\As},c^{\Bs}) \mid c \in \consts{\sg}{t}\} \\
        &\cup \{(a_j,b_j) \mid \text{ for $r_j \in P_{t}$ and $j \leq n$} \}
        \end{align*}
        where $a_j = h_{1,t}(d_{t}(r_j))$ and $b_j = h_{2,t}(d_{t}(r_j))$
        is a partial isomorphism on $J$.
    \end{itemize}

For the base case $0$, since $h_1:(W,\omega) \rightarrow \EMF{\Ek}\trJ\As$ and $h_2:(W,\omega) \rightarrow \EMF{\Ek}\trJ\Bs$ are $\sigma$-morphism, $h_{1}$ and $h_{2}$ send constants to constants. Therefore, the $S$-sorted set given by $\gamma_{s} = \{(h_1(c^{R}),h_2(c^{R})) \mid c \in \consts{\sg}{s}\} = \{(c^{\As},c^{\Bs}) \mid c \in \consts{\sg}{s}\}$ is a partial isomorphism (and in particular, a partial isomorphism on $J$). We set $P$ to be the $0$-length chain containing only interpretations for the constants in $\consts{\sg}{s}$. 

For the inductive step, suppose $s \in S$ and $b \in B_s$. By the inductive hypothesis, there exists a path $(P,\pi)$ containing constants and a chain $r_1 \sqsubset_{\pi} \dots \sqsubset_{\pi} r_n$ with embedding $d:(P,\pi) \rightarrowtail (W,\omega)$. Consider the extension of this path to $(Q,\chi)$ where $Q_{s'} = P_{s'}$ for $s' \not= s$ and $Q_{s} = P_{s} \cup \{r_{n+1}\}$. We define $\chi$ componentwise where $\chi_{s'} = \pi_{s'}$ for all $s' \not=s$ and $\chi_{s}(r) = \pi_{s}(r)$ for $r \not= r_{n+1}$, $\chi_{s}(r_{n+1}) = [r_1{:}s_1,\dots,r_n{:}s_n,r_{n+1}{:}s]$. We can define an embedding $e:(Q,\chi) \rightarrowtail \EMF{\Ek}\trJ\Bs$ such that $e_{s'} = h_{2,s} \circ d_{s'}$ for $s' \not= s$, $e_{s}(r) = h_{2,s}(d_{s}(r))$ for all $r \not= r_{n+1}$ and $e_{s}(r_{n+1}) = h_{2,s_n}(d_{s_n}(r_{n}))[b{:}s]$. By construction, we have that the following diagram commutes:
\begin{center}
\begin{tikzcd}
    (P,\pi) \ar[r,rightarrowtail] \ar[d,rightarrowtail,"e_0"'] & (Q,\chi) \ar[d,rightarrowtail,"e_1"] \\ 
    (W,\omega) \ar[r,"h_2"'] & \EMF{\Ek}\trJ\Bs
\end{tikzcd}
\end{center}

Hence, by $h_{2}$ being a open map, there exists an embedding $f:(Q,\chi) \rightarrowtail (W,\omega)$. To verify that $\gamma_{n,t}$ is partial isomorphism on $J$, we note that by $h_1,h_2$ pathwise embeddings, we have that $h_1 \circ f$ and $h_2 \circ f$ are embeddings. Therefore:
\begin{align*}
R^{\trJ\As}(a_{i_1},\dots,a_{i_m}) &\Leftrightarrow R^{\trJ\As}(\varepsilon_{\As}(h_{1,s_1}(f(r_{i_1}))),\dots,\varepsilon_{\As}(h_{1,s_m}(f(r_{i_m})))) \\
&\Leftrightarrow R^{\Ek\trJ\As}(h_{1,s_1}(f(r_{i_1})),\dots,h_{1,s_m}(f(r_{i_m}))) \\
&\Leftrightarrow R^{Q}(r_{i_1},\dots,r_{i_m}) \\
&\Leftrightarrow R^{\Ek\trJ\Bs}(h_{2,s_1}(f(r_{i_1})),\dots,h_{2,s_m}(f(r_{i_m}))) \\
&\Leftrightarrow R^{\trJ\Bs}(\varepsilon_{\Bs}(h_{2,s_1}(f(r_{i_1}))),\dots,\varepsilon_{\Bs}(h_{2,s_m}(f(r_{i_m})))) \\
&\Leftrightarrow R^{\trJ\Bs}(b_{i_1},\dots,b_{i_m})
\end{align*}
for all relations of $R \in \sg^{J}$ of arity $R$. In particular, this equivalence holds for the relations $I_{s} \in \sg^{J}$. Since $I$ is interpreted as equality in $\trJ \As$ and $\trJ \Bs$, we have that $\gamma_{n,t}$ is a partial isomorphism on $J$. 

Now to show that Duplicator has winning strategy in $\EF_{k}(\As,\Bs)$, suppose at the $n$-th round Spoiler chooses $s$ and $b \in B_s$, then by $(*)$, Duplicator can respond with $a \in A_s$ producing a partial isomorphism $\gamma_{n}$ on $J$. If Spoiler chooses $s \in S$ and $a \in A_s$, we can prove a statement similar to $(*)$ which implies that Duplicator can respond with a $b \in B_s$ that results in a partial isomorphism on $J$.
\end{proof}

Proposition~\ref{prop:bisim-game-logic} and \ref{prop:bisim-game-comonad} conclude the proof of the first item of proposition~\ref{prop:many-sorted-ek-bisim}, demonstrating that $\equiv_{\FO_k(S,J)}$ is captured by the many-sorted $\Ek$ 

For the next part, we use the following fact, which allows us to reason about $\Eknc$-coalgebras the same way as about $\Ek$-coalgebras.

\begin{lemma}
    Given a many-sorted signature $\sg$, and the comonads $\Ek$ and $\Eknc$ on $\Rs(\sg)$ and $\Rsnc(\sg)$, respectively, the categories $\EM{\Ek}$ and $\EM{\Eknc}$ are isomorphic.
    
    Moreover, the factorisation systems (surjective, strong mono) on $\Rs(\sg)$ and $\Rsnc(\sg)$ lift to the identical factorisation system on both categories of coalgebras. 
\end{lemma}
\begin{proof}
    We check that every morphism involved in the definition $A \bisim_{\Ek} B$ exists in the category $\Rnc(\sg)$. First, observe that the embedding of categories $I\colon \Rnc(\sg) \hookrightarrow \R(\sg)$ lifts to an embedding $\wh I\colon \EM{\Eknc} \hookrightarrow \EM{\Ek}$. Indeed, the obvious comonad morphism $I\circ \Eknc \to \Ek \circ I$ introduces a functor $\wh I$, mapping every coalgebra $A \to \Eknc A$ to $A \to \Eknc A \hookrightarrow \Ek A$. Moreover, $\wh I$ is surjective on objects (i.e.\ $\EM{\Eknc}$ is a wide subcategory of $\EM{\Ek}$), because every coalgebra $A \to \Ek A$ is uniquely determined by its order $\sqsubseteq_\alpha$ on non-constants via the mapping $a \mapsto [x_1, \dots, x_n]$ where $[x_1, \dots, x_n]$ is the chain of non-constant elements below $a$ (with $x_n = a$). This chain only lists non-constants and therefore factors through $\Eknc(A) \sue \Ek(A)$. Moreover, the underlying homomorphism of every coalgebra morphism in $\EM{\Ek}$ is in $\Rnc(\sg)$ (i.e.\ $\wh I$ is an isomorphism of categories).

    Next, observe the following.
    \begin{enumerate}
        \item $\Rnc(\sg)$ contains all embeddings: this is automatic by the definition of $\Rnc(\sg)$ and the fact that embeddings are injective.
        \item $\Eknc$ preserves embeddings.
        \item The factorisation system (surjective, strong injective) in $\Rel$ restricts to $\Rnc(\sg)$.
    \end{enumerate}
    By (1), (2), (3) and lemma~\ref{l:lifting-EMfs}, $\EM{\Eknc}$ consists of the same quotients and embeddings as in~$\EM{\Ek}$.
\end{proof}

For the second item of proposition~\ref{prop:many-sorted-ek-bisim} we remark that the coalgebra $(W,v)$ in the proof of proposition~\ref{prop:bisim-game-comonad} can be be restricted to $\Eknc$, i.e.\ plays where constants are never selected by Spoiler nor Duplicator. Conversely, we know that in the presence of equality, the action of Spoiler and Duplicator are determined when playing constants. So a span of open pathwise embeddings $F^{\Eknc}(\trJnc A) \leftarrow R \rightarrow F^{\Eknc}(\trJnc B)$ contains a full strategy for the $\EF^{J}_{k}(\As,\Bs)$ game. Hence, the construction proceeds as in the direction (2) $\Rightarrow$ (1) of the proof of proposition~\ref{prop:bisim-game-comonad}.

\subsection{Proofs for section~\ref{s:two-sorted-transl}}

\MSOManySortedEk*
We prove that bisimulation for the $\trE$-lifting of $S$-sorted equivalence captures equivalence in $\MSO_k$, by first noting the functor $\trE$ decomposes as $\trJ \circ \oE$ where $J = \{\fst\}$ and $\oE\colon\Rel \rightarrow \mathcal R(\sgE \setminus \{I\})$ is functor from $\Rel$ to a category of $S$-sorted $\sgE$-structures $\mathcal R(\sgE \setminus \{I\})$ without the $I$-relation. For a $\sg$-structure $\As$, $\oE(\As)$ is the $\sgE \setminus \{I\}$ reduct of $\trE(\As)$. We then prove the following lemma by adapting the translation detailed in footnote 12 of \cite{vaananen2021Stanford} which translates full second-order logic into many-sorted first-order logic:
\begin{lemma}
\label{lem:mso-to-fo-translate}
\begin{itemize}
    \item For every $\MSO_k$ sentence $\phi$ in signature $\sg$, there exists a $\FO_k(S,J)$ sentence $F(\phi)$ in signature $\sgE \setminus \{I\}$ such that
    \[ \As \vDash \phi \Leftrightarrow \oE\As \vDash F(\phi) \]
    \item For every $\FO_k(S,J)$ sentence $\phi$ in signature $\sgE \setminus \{I\}$, there exists a $\MSO_k$ sentence $M(\phi)$ in signature $\sg$ such that
    \[ \As \vDash M(\phi) \Leftrightarrow \oE\As \vDash \phi \]
\end{itemize}
\end{lemma}
\begin{proof}
We can define $F$ and $M$ by recursion on $\MSO$ and $\FO$ formulas (respectively):
\begin{align*}
    F(x_i = x_j) &= v_i{:}\fst = v_j{:}\fst\\
    F(x_i \in X_j) &= e(v_i{:}\fst, w_j{:}\snd) \\
    F(R(x_{i_1},\dots,x_{i_r})) &= R(v_{i_1}{:}\fst,\dots,v_{i_r}{:}\fst) \\
    F(\varphi \land \psi) &= F(\varphi) \land F(\psi) \\
    F(\neg \varphi) &= \neg F(\varphi) \\
    F(\exists x_i\, \varphi) &= \exists\, v_i{:}\fst\, F(\varphi) \\
    F(\exists X_i\, \varphi) &= \exists\, w_i{:}\snd\, F(\varphi)
\end{align*}
\begin{align*}
    M(v_i{:}\fst = v_j{:}\fst) &= x_i = x_j \\
    M(e(v_i{:}\fst, w_j{:}\snd)) &= x_i \in X_j \\
    M(R(v_{i_1}{:}\fst,\dots,v_{i_r}{:}\fst)) &= R(x_{i_1},\dots,x_{i_r}) \\
    M(\varphi \land \psi) &= M(\varphi) \land M(\psi) \\
    M(\neg \varphi) &= \neg M(\varphi) \\
    M(\exists\, v_i{:}\fst\, \varphi) &= \exists x_i\, M(\varphi) \\
    M(\exists\, w_i{:}\snd\, \varphi) &= \exists X_i\, M(\varphi)
\end{align*}
The proofs that $\As \vDash \phi \Leftrightarrow \oE\As \vDash F(\phi)$ and $\As \vDash M(\phi) \Leftrightarrow \oE\As \vDash \phi$ are immediate from the standard semantics of $\vDash$.
\end{proof}

With this lemma and previous results about the many-sorted $\Ek$, we obtain the proof of \ref{thm:mso-manysortedEk}
\begin{proof}[Proof of theorem \ref{thm:mso-manysortedEk}]
\begin{align*}
    \As \equiv_{\MSO_k} \Bs &\Leftrightarrow \oE \As \equiv_{\FO_k(S,J)} \oE \Bs & \text{by lemma \ref{lem:mso-to-fo-translate}}\\
    &\Leftrightarrow \trJ \oE \As \bisim_{\Ek} \trJ \oE \Bs &\text{by prop. \ref{prop:many-sorted-ek-bisim}}\\
    &\Leftrightarrow \trE \As \bisim_{\Ek} \trE \Bs & \trE = \trJ \circ \oE 
\end{align*}
Moreover, since the condition in the second item of proposition \ref{prop:many-sorted-ek-bisim} is satisfied for our choice of $J$ and $\consts{\sg}{\fst}$, we can obtain $\As \equiv_{\MSO_k} \Bs \Leftrightarrow \trEnc \As \bisim_{\Eknc} \trEnc \Bs$ using the same argument. 
\end{proof}

\subsection{Proofs for section~\ref{s:MSO-example-edge-creation}}

\EdgeCreationMSOSmoothKleisli*
\begin{proof}
    This proof proceeds similarly to the proof that $\etaRi$ is FO-$\Ek$-smooth.  
    Axioms \ref{ax:kl-law-counit} and \ref{ax:kl-law-comultiplication} follow automatically because the underlying function of $\kappa$ is the two-sorted identity function.

    To check that $\kappa$ is a smooth Kleisli law, we to check that $\kappa$ satisfies the smoothness axioms from \ref{ax:s1} and \ref{ax:s2p} from lemma~\ref{l:s2-multilin}. 
    
    For \ref{ax:s1}, we can repeat the verification of \ref{ax:s1} as in lemma \ref{l:etaRi-FO-Ek-smooth} nearly verbatim, but instead of using the relations in $R,P_{i_z} \in \sg$, we use the corresponding relations $R,P_{i_z} \in \sgE$. By definition, the interpretations of these relations are the same.  

    For \ref{ax:s2p}, suppose $\pi\colon \Ps \to \Ek \Ps$ is a path $r_1 \sqsubset_\pi r_2 \sqsubset_\pi \dots \sqsubset_\pi r_n$ and $f\colon \Pc \to [\Ac]$ is a multilinear map. The proof that $\lift{f}$ has a minimal decomposition follows nearly verbatim to the $\FO$ case in lemma \ref{l:etaRi-FO-Ek-smooth}. The key difference is that $\kappa$ is the two-sorted identity function. 

    In addition to demonstrating that $\kappa$ is a Kleisli law. We must also show there exists a natural transformation $\theta:\etaRiE\circ\trE \rightarrow \trE\circ\etaRi$ such that $\EMF{\Ek}(\theta_\As)$ is an open pathwise embedding. The components $\theta_{\As}$ are given by $S$-sorted functions where $\theta_{\As,t}$ sends an element $x \in \As_t \in \etaRiE(\trE(\As))$ the same element in $\trE\etaRi(\As)$. The interpretation of $Q \in \sgE \setminus \{e,I\}$  has the same tuples as the interpretation of $Q \in \sg$ (including $R$). Moreover, since $\trE\etaRi(\As)$ and $\etaRiE\trE(\As)$ have the same universe, the membership relation (consequently, interpretations of $e$) and the identity relation (consequently, interpretations of $I$) are the same, so we have that $\theta_{\As}$ is a $\sgE$-morphism. Since $\theta_{\As}$ is the $S$-sorted identity function, the underlying set function of $\EMF{\Ek}(\theta_\As)$ is the identity and it is easy to check that $\EMF{\Ek}(\theta_\As)$ is an open pathwise embedding. 
\end{proof}

\subsection{Proofs for section~\ref{s:twisted-sums}}

We first verify the following claim.

\begin{lemma}
$\theta \colon \trE(\As) \bowtie \trE(\Bs) \to \trE(\As \plus \Bs)$ is a strong onto homomorphism.
\end{lemma}
\begin{proof}
    Since $\theta$ is defined as the identity function on the $\fst$-sort, this holds true for the $\sg_3$-reduct and the $\fst$-sort part of $\theta$. On the $\snd$-sort, $\theta$ is clearly onto. So we only need to check that the $e$ relation is reflected. Let $a\in A$, $M\sue A$ and $N \sue B$. Assume that $(\inc(a),\inc(M)\cup\inc(N))\in e^{\trE(\As \plus \Bs)}$, i.e.\  $\inc(a) \in \inc(M)\cup\inc(N)$. If $a$ is non-constant then clearly $a\in M$ and so $(a,(M,N))\in e^{\trE(\As) \bowtie \trE(\Bs)}$. Otherwise, $a$ is a constant, i.e.\ $a = c^\As$ for some $c\in \const{\sg}$. Then, no matter if $a\in \inc(M)$ or $a\in \inc(N)$, we have that $(c^{\As \plus \Bs}, (M,N))\in e^{\trE(\As) \bowtie \trE(\Bs)}$, which concludes the proof because $\iota(a) = \iota(c^\As) = c^{\As \plus \Bs}$.
\end{proof}

The fact that $F^{\Eknc}(\theta)$ is component-wise an open pathwise embedding then follows from the following.

\begin{lemma}
    For a multi-sorted signature $\sg$, if $q$ is a strong onto homomorphism in $\Rsnc(\sg)$ then $F^{\Eknc}(q)$ is an open pathwise-embedding.
\end{lemma}
\begin{proof}
    In the following, we will write $U \dashv F$ for the adjunction $U^{\Eknc} \dashv F^{\Eknc}$. Observe that the correspondence
    \begin{align}
        \hom(U(\pi), \As) \cong \hom(\pi, F(\As)),
        \label{eq:adj-paths}
    \end{align}
    for $\pi$ an $\Eknc$-path and $\As$ a (multi-sorted) $\sg$-structure, restrict to a bijection between strong homomorphisms $U(\pi) \to \As$ and (path) embeddings $\pi \to F(\As)$.

    To check that $F(q)\colon F(\As) \to F(\Bs)$ is a pathwise-embedding, let $e\colon \pi \embed F(\As)$ be a path embedding. By the above, this embedding corresponds to a strong homomorphism $e^r\colon U(\pi) \embed \As$. Then, $q\circ e^r$ is a strong homomorphism $e^r\colon U(\pi) \embed \Bs$ because $q$ is also strong. Therefore, the corresponding coalgebra morphism $(q\circ e^r)^l\colon \pi \embed F(\Bs)$ is an embedding and $F(q)\circ e = (q\circ e^r)^l$ by naturality of \eqref{eq:adj-paths}.

    For openness, let $\pi$ and $\rho$ be paths and the following is a commutative square of path embeddings shown on left below.
    \[
        \begin{tikzcd}
            \pi\dar[>->,swap]{e}\rar[>->]{h} & \rho\dar[>->]{g} \\
            F(A) \rar{F(q)} & F(B)
        \end{tikzcd}
        \qquad
        \qquad
        \begin{tikzcd}
            U(\pi)\dar[swap]{e^l}\rar{U(h)} & U(\rho)\dar{g^l}\ar[dashed,swap]{dl}{d} \\
            A \rar[->>]{q} & B
        \end{tikzcd}
    \]
    In view of \eqref{eq:adj-paths}, we have a commuting diagram on the right above, giving us that $\pi$ corresponds to a word (of non-constants) $[x_1:s_1,\dots,x_n:s_n]$ in $A$ and $\rho$ corresponds to a word (of non-constants) $[y_1:s'_1,\dots,y_m:s'_m]$ where $n \leq m$, $s_i = s'_i$ and $q(x_i) = y_i$, for $1\leq i \leq n$. Since $q$ is onto, there exist $x_j \in A_{s_j}$ such that $q_{s_j}(x_j) = y_j$, for $n+1 \leq j \leq m$. We define $d\colon U(\rho) \to A$ by sending $y_i \mapsto x_i$ (for $1 \leq i \leq m$). It is easy to see that $d$ is a homomorphism because $q$ is a strong homomorphism. Also, the map $d$ makes both triangles in the diagram commute, by definition. Consequently, we have a coalgebra homomorphism $d^r\colon \rho \to F(A)$ making the obvious triangles on the left diagram above commute.
\end{proof}

\begin{lemma}
    $\kappa\colon \Eknc(\As \bowtie \Bs) \to \Eknc(\As) \bowtie \Eknc(\Bs)$ is a smooth Kleisli law.
\end{lemma}
\begin{proof}
    The proof of the fact that $\kappa$ is a well-defined natural transformation that satisfies \ref{ax:kl-law-counit}, \ref{ax:kl-law-comultiplication}, \ref{ax:s1}, and \ref{ax:s2p} is identical to the single-sorted proof of the same, in the proof lemma~\ref{l:kappa-coproducts} above, with the only addition that we have to keep track of sorts.

    Also, $\EM{\Eknc}$ has equalisers of reflexive pairs. This follows from lemma~\ref{l:EM-equalisers}, adapted verbatim to the multi-sorted setting. Since $\kappa$ is a Kleisli law, the functor $\bowtie$ lifts to a functor $\blift{\bowtie}\colon \EM{\Eknc}\times \EM{\Eknc} \to \EM{\Eknc}$ and so $\kappa$ is smooth by theorem~\ref{t:smoothness} and lemma~\ref{l:s2-multilin}.
\end{proof}

This now follows from the observations above.
\plusMSOsmooth*

\subsection{Proofs for section~\ref{s:courcelle}}

\ComonadicCourcelle*
\begin{proof}
    For the proof we build a (bottom-up finite) tree automaton. The set of states of this automaton is the set $S = \{1,\, \dots,\, (t_1 + \dots + t_n)\}$ where $t_i$ is the total number of $\bisim_{\C_i}^{\tr_i}$-equivalence classes in $\R_0(\sg_i)$, for $i=1,\dots, n$. Since every comonad $\C_i$ has finite MSO type and since we only have finitely many categories, we know that $t_1 + \dots + t_n$ is a finite number. For the alphabet we take the set $\{ o_\op \mid \op \in \Ops\}$ of names for operations in $\Ops$, where the arity of $o_\op$ is precisely the arity of the functor $\op$ (ignoring the types).

    Next, let $l_A \in S$ be the label corresponding to the $\bisim_{\C_i}^{\tr_i}$-equivalence class of $A$, for every $A\in \R(\sg_i)$ and $1 \leq i \leq n$. Then, for transition rules, given $\op\in \Ops$ of type $\R(\sg_{i(1)})\times \dots \times \R(\sg_{i(u)}) \to \R(\sg_{i(u+1)})$ and structures $A_l\in \sg_{i(l)}$, for $1 \leq l \leq u$, we add the transition
    \[ o_\op(l_{A_1},\dots,l_{A_n}) \ee\leadsto l_{\op(A_1,\dots,A_n)}. \]
    We see that this transition is deterministic because every $\op\in \Ops$ is MSO-$(\C_1,\dots,\C_n)$-smooth.

    Lastly, we define the accepting states as those $l_A\in S$ such that the $\tau$-reduct of $A$ is in $\Delta$. This, is well-defined by assumption (3).

    Now, given any term $t_A$, we build a tree $T_A$ by translating every constant $B\in \Gens$ into the label $l_B$ and every operation application of $\op \in \Ops$ into the letter $o_\op$. Recall that checking if an input tree is accepted can be done in linear time, for any tree automaton. Lastly, observe that, by construction, $T_A$ is accepted if and only if $A\in \Delta$.
\end{proof}

\subsection{Proofs for section~\ref{s:tw-thm}}
Recall that $\forgCD\colon \R_0(\sg_C) \to \R_0(\sg_D)$, for $D \sue C \sue \{c_1,\dots,c_m\}$,
is the operation which sends a relational structure $\As$, in signature $\sg_C$, to its $\sg_D$-reduct. This operation naturally extends to the two-sorted functor $\forgCD\E\colon \Rnc(\sgE_C) \to \Rnc(\sgE_D)$, which acts as the identity mapping on morphisms.

We immediately see that $\trE \circ \forgCD \cong \forgCD\E \circ \trE$. Therefore, for $\MSO$-$\Eknc$-smoothness of $\forgCD$, we only need to find a smooth Kleisli law
\[ \kappa\colon \Eknc \circ \forgCD\E  \to \forgCD\E \circ \Eknc. \]
Given $\As$ in $\Rnc(\sgE_C)$, we divide its universe as $A_\fst\uplus A_\snd = A_0 \uplus C_0 \uplus D$ where $D = \{ c^\As \mid c\in \const{\sg_D} \}$, $C_0 = \{ c^\As \mid c\in \const{\sg_C} \setminus \const{\sg_D} \}$ and $A_0 = A_\fst \uplus A_\snd \setminus (D \cup C_0)$.
Recall that the universe of $\Eknc(\forgCD\E(\As))$ consists of $(A_0\cup C_0)^{\leq k}\uplus \const{\sg_D}$ whereas the universe of $\forgCD\E(\Eknc(\As))$ consists of $A_0^{\leq k}\uplus \const{\sg_C}$, with only $\const{\sg_D}\sue \const{\sg_C}$ marked as constants.
Define $\kappa$ componentwise on non-constants by (while preserving the sorts as usual):
\begin{align*}
    \kappa\colon \Eknc(\forgCD\E(\As)) &\ee\longrightarrow \forgCD\E(\Eknc(\As)), \\[0.5em]
     w\in (A_0 \cup C_0)^{\leq k}
    &\ee\longmapsto
    \begin{cases}
        \ex(w) \in A_0^{\leq k} & \text{ if } \counit(w) \in A_0 \\[0.5em]
        c^{\Eknc(\As)} & \text{ if } \counit(w) = c^\As \in C_0
    \end{cases}
\end{align*}
Here $\ex([x_1 : s_1, \dots, x_n : s_n])$ is defined as $[ x_i:s_i \mid x_i \in A_0]$. It is clear that $\kappa$ is a $\sg_D$-structure homomorphism. Furthermore, it sends non-constants to non-constants since $c^{\Eknc(\As)}$ with $c^\As \in C_0$ is no longer a constant in $\forgCD\E(\Eknc(\As))$. The reasoning justifying that axioms \ref{ax:kl-law-counit} and \ref{ax:kl-law-comultiplication} hold is exactly the same for lemma~\ref{l:kappa-coproducts} and \ref{ax:s1} holds for trivial reasons. Lastly, we check \ref{ax:s2p}. Assume we are given a path embedding, in terms of a multilinear map~$e$ as shown below.
\[
    \begin{tikzcd}[->]
        \Ps \arrow[rr, "e"] \dar[swap]{\pi} & & \forgCD\E(\As) \dar{\forgCD\E(\alpha)} \\
        \Eknc(\Ps) \rar[swap]{\Eknc(e)} & \Eknc\forgCD\E(\As) \rar[swap]{\kappa} & \forgCD\E(\Eknc \As)
    \end{tikzcd}
\]
Let $x_1 : s_1 \sqsubset_\pi x_1 : s_1 \sqsubset_\pi \dots \sqsubset_\pi x_n : s_n$ be the chain of non-constant elements of $\Ps$. The commutativity of the above square means that, for every $i$,
\[ \kappa([a_1 : s_1, \dots, a_i : s_i]) = \alpha(a_i) \]
where $a_j$ is set to be $e(x_j)$, for $1 \leq j \leq n$. In other words, if $a_i\in A_0$ it must be that the word $\kappa([a_1 : s_1, \dots, a_i : s_i])$ is the chain of non-constant elements $\sqsubset_\alpha$-below $a_i$ in $\As$. If $a_i \in C_0$ then we automatically get that $\kappa([a_1 : s_1, \dots, a_i : s_i]) = a_i$ and also $\alpha(a_i) = a_i$, making this automatically satisfied. Therefore, the word $[a_1 : s_1, \dots, a_n : s_n]$ consists of elements from a path $(\Ps', \pi')$ in $(\As,\alpha)$ interleaved by constant elements from $C_0$. It is immediate that $e$ factors through $\forgCD\E(e')$ of the embedding $e'\colon (\Ps', \pi') \embed (\As,\alpha)$ and that this $(\Ps', \pi')$ is minimal possible. As a result we obtain the following by theorem~\ref{t:smoothness} and lemma~\ref{l:s2-multilin}.

\begin{lemma}
    For every $D \sue C \sue \{c_1,\dots,c_m\}$, the operation $\forgCD\colon \R(\sg_C) \to \R(\sg_D)$ is $\MSO$-$\Eknc$-smooth.
\end{lemma}

\subsection{Proofs for section~\ref{s:cw-thm}}

What remains to check to support theorem~\ref{t:cw-courcelle-recovered} is the following claim.
\begin{restatable}{lemma}{rhoijMSOEkSmooth}
For every $i,j \in \setm$, the operation $\rhoij:\R(\sgm) \to \R(\sgm)$ is $\MSO$-$\Ek$-smooth.
\end{restatable}
\begin{proof}
It is clear that $\rhoij\colon\Relk \rightarrow \Relk$ and the corresponding two-sorted operation $\rhoij\E\colon\RelEm \rightarrow \RelEm$ are functorial. We need to define smooth Kleisli law 
    \[ \kappa\colon \Ek\circ\rhoij\E \rightarrow \rhoij\E\circ\Ek \]
The components $\kappa_{\As}$ are given by $S$-sorted functions where $\kappa_{\As,t}$ sends a word $w \in (\Ek(\rhoij(A)))_t$ to the same word $w \in (\rhoij(\Ek(A)))_t$ for $t \in \{\fst,\snd\} = S$. This is the same definition, on the level of $S$-sorted sets, as the smooth Kleisli law used to witness that $\etaRi$ is $\MSO$-$\Ek$-smooth. We also need to exhibit a natural transformation $\theta:\rhoij\E \circ \trE \rightarrow \trE \circ \rhoij$. As in the case of $\etaRi$, we define the components $\theta_\As$ to be the $S$-sorted identify function sending an element $a \in \rhoij\E(\trE(\As))$ to the same element $a \in \trE(\rhoij\As)$. Consequently, the proof of lemma \ref{l:etaRi-MSO-Ek-smooth} can be adapted to demonstrate that $\rhoij$ is $\MSO$-$\Ek$-smooth.
\end{proof}

\end{document}